\documentclass[preprint,11pt]{article}
\usepackage{fullpage}
\usepackage{multirow}
\usepackage{makecell}
\usepackage{array}
\usepackage{amssymb,amsfonts,amsmath,amsthm,amscd,dsfont,mathrsfs}
\usepackage{graphicx,float,psfrag,epsfig,amssymb}
\usepackage[usenames,dvipsnames,svgnames,table]{xcolor}
\definecolor{darkgreen}{rgb}{0.0,0,0.9}
\usepackage[pagebackref,letterpaper=true,colorlinks=true,pdfpagemode=none,citecolor=OliveGreen,linkcolor=BrickRed,urlcolor=BrickRed,pdfstartview=FitH]{hyperref}
\usepackage{wrapfig}
\usepackage{relsize}
\usepackage{color}
\usepackage{pict2e}
\usepackage[tight]{subfigure}
\usepackage{algorithm}
\usepackage[noend]{algorithmic}
\usepackage{caption}
\usepackage{nameref}

\definecolor{Red}{rgb}{1,0,0}
\definecolor{Blue}{rgb}{0,0,1}
\definecolor{Olive}{rgb}{0.41,0.55,0.13}
\definecolor{Green}{rgb}{0,1,0}
\definecolor{MGreen}{rgb}{0,0.8,0}
\definecolor{DGreen}{rgb}{0,0.55,0}
\definecolor{Yellow}{rgb}{1,1,0}
\definecolor{Cyan}{rgb}{0,1,1}
\definecolor{Magenta}{rgb}{1,0,1}
\definecolor{Orange}{rgb}{1,.5,0}
\definecolor{Violet}{rgb}{.5,0,.5}
\definecolor{Purple}{rgb}{.75,0,.25}
\definecolor{Brown}{rgb}{.75,.5,.25}
\definecolor{Grey}{rgb}{.5,.5,.5}
\definecolor{Pink}{rgb}{1,0,1}
\definecolor{DBrown}{rgb}{.5,.34,.16}

\definecolor{Black}{rgb}{0,0,0}

\DeclareMathAlphabet{\mathpzc}{OT1}{pzc}{m}{it}

\newtheorem{propo}{Proposition}[section]
\newtheorem{lemma}[propo]{Lemma}
\newtheorem{definition}[propo]{Definition}
\newtheorem{coro}[propo]{Corollary}
\newtheorem{thm}[propo]{Theorem}

\newtheorem{remark}[propo]{Remark}

\def\tlambda{\widetilde{\lambda}}
\def\tZ{\widetilde{Z}}
\def\hT{\widehat{T}}

\def\cF{{\cal F}}

\def\cC{{\cal C}}

\def\cE{{\cal E}}

\def\reals{{\mathbb R}}

\def\eps{{\varepsilon}}
\def\prob{{\mathbb P}}
\def\E{{\mathbb E}}
\def\Var{{\rm Var}}

\def\L0{{L_0}}

\def\de{{\rm d}}
\def\<{\langle}
\def\>{\rangle}

\def\bX{{\mathbf X}}

\def\htheta{\widehat{\theta}}
\def\hSigma{\widehat{\Sigma}}
\def\hsigma{\widehat{\sigma}}

\def\supp{{\rm supp}}

\def\F{{\sf F}}
\def\ind{{\mathbb I}}

\def\F{{\sf F}}
\def\normal{{\sf N}}

\def\sT{{\sf T}}

\def\id{{\rm I}}

\def\avglength{{\rm Avglength}}

\def\cov{\widehat{\sf Cov}}
\def\hprob{\widehat{\prob}}

\def\event{\mathcal{E}}

\def\v*{v_0}
\def\T*{T_0}

\def\u*{u_0}
\def\F*{F_0}

\definecolor{olivegreen}{rgb}{0,0.6,0.4}

\def\dist{\overset{{\rm d}}{\longrightarrow}} 
\def\tW{\widetilde{W}}

\def\bias{{\sf Bias}}

\def\coh{{\mu}}
\def\com{{\mu}_{\rm min}}
\def\lb{{\gamma}}
\def\coev{{\cal G}}
\def\re{\phi_{\rm RE}}
\def\cB{{\cal B}}

\def\FWER{{\rm FWER}}
\def\hTf{\widehat{T}^{{\rm F}}}

\newcommand{\ajcomment}[1]{}

\makeatletter
\newcommand{\labitem}[2]{%
\def\@itemlabel{\text{#1}}
\item
\def\@currentlabel{#1}\label{#2}}
\makeatother

\addtocontents{toc}{\protect\setcounter{tocdepth}{2}}

\title{Confidence Intervals and Hypothesis Testing for\\
  High-Dimensional Regression}

\author{Adel Javanmard
            \footnote{Department of Electrical Engineering, Stanford University. Email: \url{adelj@stanford.edu} }
             \,and Andrea~Montanari 
            \footnote{Department of Electrical Engineering and Department of Statistics, Stanford University. Email: \url{montanar@stanford.edu}}
            }

\begin{document}

\maketitle

\begin{abstract}
Fitting high-dimensional statistical models often requires the use
of non-linear parameter estimation procedures. As a consequence, it is
generally impossible to obtain an exact characterization of the probability
distribution of the parameter estimates. This in turn implies that it
is extremely challenging to quantify the \emph{uncertainty} associated
with a certain parameter estimate. Concretely, no commonly accepted
procedure exists for computing classical measures of uncertainty and
statistical significance as confidence intervals or $p$-values for these models.

We consider here high-dimensional linear regression problem, and propose 
an efficient algorithm  for constructing confidence intervals and $p$-values. 
The resulting confidence intervals have nearly optimal size. When
testing for the null hypothesis that a certain parameter is vanishing,
our method has nearly optimal power.

Our approach is based on constructing a `de-biased' version of
regularized M-estimators. The new construction  improves
over recent work in the field in that it does not assume a special
structure on the design matrix. 
We test our method on synthetic data and a high-throughput genomic data set 
about riboflavin production rate, made publicly available by~\cite{BuhlmannBio}.
\end{abstract}

\section{Introduction}

It is widely recognized that modern statistical problems are
increasingly high-dimensional, i.e. require estimation of more
parameters than the number of observations/samples. Examples abound
from signal processing \cite{CSMRI}, to genomics \cite{peng2010regularized}, collaborative
filtering \cite{KBV09} and so on.  A number of successful estimation
techniques have been developed over the last ten years to tackle these
problems. A  widely applicable approach consists in optimizing a suitably regularized 
likelihood function. Such estimators are, by necessity, non-linear
and non-explicit (they are solution of certain optimization problems).

The use of non-linear parameter estimators
comes at a price. In general, it is impossible to characterize the
distribution of the estimator. This situation is very different from
the one of classical statistics in which either exact characterizations are
available, or asymptotically exact ones can be derived from large
sample theory \cite{van2000asymptotic}. This has an important and very concrete
consequence. In classical statistics, generic and well accepted 
procedures are available for characterizing the uncertainty associated
to a certain parameter estimate in terms of
confidence intervals or $p$-values \cite{wasserman2004all,lehmann2005testing}. However, no analogous procedures
exist in high-dimensional statistics.

In this paper we develop a computationally efficient procedure for constructing confidence
intervals and $p$-values for a broad class of high-dimensional regression
problems. The salient features of our procedure are:
\begin{enumerate}
\item[$(i)$] Our approach guarantees nearly optimal confidence interval
sizes and testing power. 
\item[$(ii)$] It is the first one to achieve this goal
under essentially no  assumptions beyond the standard conditions for high-dimensional
consistency. 
\item[$(iii)$] It allows for a streamlined analysis with
respect to earlier work in the same area.
\end{enumerate}
For the sake of clarity, we will focus our presentation on the case of
linear regression, under Gaussian noise. Section \ref{sec:NonGaussian} 
provides a detailed study of the case of non-Gaussian noise.
A preliminary report on our results was presented in NIPS 2013
\cite{ConfidenceNIPS2013}, which also discusses generalizations of the
same approach to generalized linear models, and regularized
maximum likelihood estimation.

In a linear regression model, we  are given $n$ i.i.d. pairs 
$(Y_1,X_1), (Y_2,X_2), \dots, (Y_n,X_n)$, with vectors
$X_i \in \reals^p$ and response variables $Y_i$ given by
\begin{eqnarray}\label{eqn:regression}
Y_i \,=\, \<\theta_0,X_i\> + W_i\, ,\;\;\;\;\;\;\;\; W_i\sim
\normal(0,\sigma^2)\, .
\end{eqnarray}
Here $\theta_0 \in \reals^p$ and $\<\,\cdot\,,\,\cdot\,\>$ is the standard scalar product in
$\reals^p$. 
In matrix form,
letting  $Y = (Y_1,\dotsc,Y_n)^\sT$ and denoting by $\bX$ the design matrix with
rows $X_1^\sT,\dotsc, X_n^\sT$, we have
\begin{eqnarray}\label{eq:NoisyModel}
Y\, =\, \bX\,\theta_0+ W\, ,\;\;\;\;\;\;\;\; W\sim
\normal(0,\sigma^2 \id_{n\times n})\, .
\end{eqnarray}
The goal is to estimate the unknown (but fixed) vector of parameters $\theta_0 \in \reals^p$.

In the classic setting, $n\gg p$ and the estimation method of choice is ordinary
least squares yielding $\htheta^{\rm OLS} =
(\bX^{\sT}\bX)^{-1}\bX^{\sT}Y$. In particular $\htheta^{\rm OLS}$ is Gaussian
with mean $\theta_0$  and covariance $\sigma^2 (\bX^{\sT}\bX)^{-1}$. 
This directly allows to construct confidence intervals\footnote{For
instance, letting $Q\equiv (\bX^{\sT}\bX/n)^{-1}$, $\htheta^{\rm
  OLS}_i-1.96\sigma\sqrt{Q_{ii}/n},\htheta^{\rm OLS}_i+1.96\sigma\sqrt{Q_{ii}/n}]$ is
a $95\%$ confidence interval \cite{wasserman2004all}.}.

In the high-dimensional setting where $p>n$, the matrix $(\bX^{\sT}\bX)$ is
rank deficient and one has to resort to biased estimators.
A particularly successful approach is the  LASSO~\cite{Tibs96,BP95} which
promotes sparse reconstructions through an $\ell_1$ penalty:
\begin{align}
\htheta^{n}(Y,\bX;\lambda) \equiv \arg\min_{\theta\in\reals^p}
\Big\{\frac{1}{2n}\|Y-\bX\theta\|^2_2+\lambda\|\theta\|_1\Big\}\, . \label{eq:LASSOEstimator}
\end{align}
In case the right hand side has more than one minimizer, one of them
can be selected arbitrarily for our purposes.
We will often omit the arguments $Y$, $\bX$, as they are clear from
the context. 

We denote by $S \equiv \supp(\theta_0)$ the support of
$\theta_0 \in \reals^p$, defined as
$$\supp(\theta_0) \equiv \{i\in [p]:\, \theta_{0,i} \neq 0\}\,,$$
where we use the notation $[p] = \{1,\dotsc, p\}$.
We further let $s_0\equiv |S|$. A copious theoretical literature \cite{CandesTao,BickelEtAl,buhlmann2011statistics}
shows that, under suitable assumptions on $\bX$, the 
LASSO is nearly as accurate as if the support $S$ was known \emph{a
  priori}. Namely, for $n= \Omega(s_0\log p)$, we have  $\|\htheta^n-\theta_0\|_2^2 = O(s_0\sigma^2(\log p)/n)$.

\begin{algorithm}[t]
\caption*{{\bf Table 1:} Unbiased estimator for $\theta_0$ in high-dimensional linear regression models}
\begin{algorithmic}[1]

\REQUIRE Measurement vector $y$, design matrix $\bX$, parameters $\lambda$, $\coh$.

\ENSURE Unbiased estimator $\htheta^u$.

\STATE Let $\htheta^n=\htheta^n(Y,\bX;\lambda)$ be the LASSO estimator as per Eq.~\eqref{eq:LASSOEstimator}.

\STATE Set $\hSigma \equiv (\bX^\sT \bX)/n$.

\FOR{$i = 1, 2, \dotsc, p$} 

\STATE  Let $m_i$ be a solution of the convex program:
\begin{eqnarray}
\begin{split}\label{eq:optimization}
&\text{minimize } \quad \, m^\sT \hSigma m\\
&\text{subject to} \quad \|\hSigma m - e_i \|_{\infty} \le \coh\,,
\end{split}
\end{eqnarray}
where $e_i \in \reals^p$ is the vector with one at the $i$-th position and zero everywhere else.
\ENDFOR

\STATE Set $M = (m_1,\dotsc,m_p)^\sT$. If any of the above problems is not feasible, then set $M= \id_{p\times p}$. 

\STATE Define the estimator $\htheta^u$ as follows:
\begin{eqnarray}
\htheta^u = \htheta^n(\lambda) + \frac{1}{n}\, M \bX^\sT(Y - \bX \htheta^n(\lambda)) \label{eq:hthetau}
\end{eqnarray}
\end{algorithmic}
\end{algorithm}
As mentioned above, these remarkable properties come at a price. Deriving an exact
characterization for the distribution of $\htheta^n$ is not tractable in general, and hence there is no simple procedure to construct
confidence intervals and $p$-values. A closely related property is that
$\htheta^n$ is biased, an unavoidable property in high dimension,
since a point estimate  $\htheta^n\in\reals^p$ must be produced from
data in lower dimension $Y\in\reals^n$, $n<p$. We refer to
Section \ref{sec:BiasDiscussion} for further discussion of this point.

In order to overcome this challenge, we construct a de-biased
estimator from the LASSO solution.
The de-biased estimator is given by the simple formula $\htheta^u = \htheta^n + (1/n)\, M \bX^\sT(Y - \bX
\htheta^n)$, as in Eq. (\ref{eq:hthetau}). The basic intuition is
that  $\bX^\sT(Y - \bX \htheta^n)/(n\lambda)$ is  a subgradient
of the $\ell_1$ norm at the LASSO solution $\htheta^n$. By
adding a term proportional to this subgradient, our procedure
compensates the bias introduced by the $\ell_1$ penalty in the LASSO.

We will prove in Section \ref{sec:Debiased} that $\htheta^u$ is
approximately Gaussian, with mean $\theta_0$ and covariance
$\sigma^2(M\hSigma M)/n$, where $\hSigma = (\bX^{\sT}\bX/n)$ is
the empirical covariance of the feature vectors. This result allows to
construct confidence intervals and $p$-values in complete analogy with
classical statistics procedures. For instance, letting $Q \equiv M\hSigma M$,
$[\htheta^u_i-1.96\sigma\sqrt{Q_{ii}/n},
\htheta^u_i+1.96\sigma\sqrt{Q_{ii}/n}]$ is a $95\%$ confidence
interval. The size of this interval is of order
$\sigma/\sqrt{n}$, which is the optimal (minimum) one, i.e. the same
that would have been obtained by knowing \emph{a priori} the support
of $\theta_0$. In practice the noise standard deviation is
not known, but $\sigma$ can be replaced by any consistent estimator
$\hsigma$ (see Section \ref{sec:Inference} for more details on this).

A key role is played by the matrix $M \in \reals^{p\times p}$ whose function  is to
`decorrelate' the columns of $\bX$. We propose here to construct
$M$ by solving a convex program that aims at optimizing two
objectives. One one hand, we try to control 
$|M \hSigma - \id |_{\infty}$ (here and below $|\,\cdot\,|_{\infty}$
denotes the entrywise $\ell_{\infty}$ norm) which --as shown in
Theorem \ref{thm:main_thm}-- controls the
non-Gaussianity and bias of $\htheta^u$. On the other, we minimize
$[M\hSigma M]_{i,i}$, for each $i\in [p]$, which controls the variance of
$\htheta_{i}^u$. 


The idea of constructing a de-biased estimator of the form $\htheta^u = \htheta^n + (1/n)\, M \bX^\sT(Y - \bX
\htheta^n)$ was used by the present authors in
\cite{javanmard2013hypothesis}, that suggested the choice $M=c\Sigma^{-1}$, with $\Sigma =
\E\{X_1X_1^{\sT}\}$ the population covariance matrix and $c$ a positive constant. A simple estimator
for $\Sigma$ was proposed for sparse covariances, but asymptotic
validity and optimality were proven only for uncorrelated Gaussian designs
(i.e. Gaussian $\bX$ with $\Sigma = \id$). 
Van de Geer, B\"ulhmann, Ritov and Dezeure \cite{GBR-hypothesis} used the same
construction with $M$ an estimate of $\Sigma^{-1}$ which is appropriate for
sparse inverse covariances. These authors prove semi-parametric
optimality in a non-asymptotic setting, provided the sample size is at
least $n = \Omega((s_0\log p)^2)$.

From a technical point of view, our proof starts from a simple decomposition
of the de-biased estimator $\htheta^u$ into a Gaussian part and an
error term, already used in \cite{GBR-hypothesis}.
However --departing radically from earlier work-- we realize that $M$ need not be a
good estimator of $\Sigma^{-1}$ in order for the de-biasing procedure to work. 
We instead set $M$ as to minimize the error term and the variance of
the Gaussian term.
As a consequence of this choice, our approach applies to general 
covariance structures $\Sigma$. By contrast, earlier approaches applied only
 to sparse $\Sigma$, as in \cite{javanmard2013hypothesis}, or
sparse $\Sigma^{-1}$ as in \cite{GBR-hypothesis}. 
The only assumptions we
make on $\Sigma$ are the standard compatibility conditions
required for high-dimensional consistency
\cite{buhlmann2011statistics}.
A detailed comparison of our results with the ones of
\cite{GBR-hypothesis} can be found in Section
\ref{sec:GBR-comparison}.

Our presentation is organized as follows.
\begin{description}
\item[Section \ref{sec:DebiasedGen}] considers a general debiased
  estimator of the form $\htheta^u = \htheta^n + (1/n)\, M \bX^\sT(Y - \bX
\htheta^n)$. We introduce a figure of merit of the pair $M,\bX$, termed the
generalized coherence parameter $\coh_*(\bX;M)$. We show that, if the
generalized coherence is small, then
the debiasing procedure is effective (for a given deterministic
design), see Theorem \ref{thm:deterministic}.

We then turn to random designs, and show that the generalized
coherence parameter can be made as small as $\sqrt{(\log p)/n}$,
though a convex optimization procedure for computing $M$. This results
in a bound on the bias of $\htheta^u$, cf. Theorem \ref{thm:main_thm}: the largest entry of the bias is of order
$(s_0\log p)/n$. This must be compared with the standard deviation of
$\htheta^u_i$, which is of order $\sigma/\sqrt{n}$.
The conclusion is that, for $s_0 = o(\sqrt{n}/\log p)$, the bias of
$\htheta^u$ is negligible.
\item[Section \ref{sec:Inference}] applies these distributional
  results to deriving confidence intervals and hypothesis testing
  procedures for low-dimensional marginals of $\htheta_0$. The basic
  intuition is that $\htheta^u$ is approximately
  Gaussian with mean $\theta_0$, and known covariance structure. Hence
  standard optimal tests can be applied. 

We prove a general lower bound on the power of our testing procedure,
in Theorem \ref{thm:error-power}.
In the special case of Gaussian random designs with i.i.d. rows, we
can compare this with the upper bound proved in
\cite{javanmard2013hypothesis},
cf. Theorem \ref{thm:GeneralUpperBound}. As a consequence, the
asymptotic efficiency of our approach is constant-optimal. Namely, it
is lower bounded by a constant $1/\eta_{\Sigma,s_0}$ which is bounded
away from $0$, cf. Theorem \ref{thm:optimality}. (For instance
$\eta_{\id,s_0}=1$, and $\eta_{\Sigma,s_0}$ is always upper bounded by
the condition number of $\Sigma$.)
\item[Section \ref{sec:NonGaussian}] uses the a central limit theorem
  for triangular arrays to generalize the above results to
  non-Gaussian noise.
\item[Section \ref{sec:simulation}] illustrates the above results through
  numerical simulations both on synthetic and on real data.
\end{description}
Note that our proofs  require stricter sparsity $s_0$ (or larger sample
size $n$) than required for consistent estimation. We assume $s_0 =
o(\sqrt{n}/\log p)$ instead of $s_0 = o(n/\log p)$
\cite{Dantzig,BickelEtAl,buhlmann2011statistics}.
The same assumption is made in \cite{GBR-hypothesis}, on top of
additional assumptions on the sparsity of $\Sigma^{-1}$.

It is currently an open question  whether successful hypothesis
testing can be performed under the weaker assumption  $s_0 = o(n/\log
p)$. We refer to  \cite{javanmard2013nearly} for preliminary work in
that direction.
The barrier at $s_0 =o(\sqrt{n}/\log p)$ is possibly related to an
analogous assumption that arises in Gaussian graphical models
selection \cite{ren2013asymptotic}.

\subsection{Further related work}

The theoretical literature on high-dimensional statistical models is
vast and rapidly growing.
Estimating sparse linear regression models is the most studied
problem in this area, and a source of many fruitful ideas. 
Limiting ourselves to linear regression, earlier work investigated prediction error
\cite{GreenshteinRitov}, model selection properties
\cite{MeinshausenBuhlmann,zhao,Wainwright2009LASSO,candes2009near}, $\ell_2$
consistency \cite{CandesTao,BickelEtAl}. .
Of necessity, we do not provide a complete set of references, and instead
refer the reader to \cite{buhlmann2011statistics} for an in-depth introduction to this area.

The problem of quantifying statistical significance in high-dimensional
parameter estimation 
is, by comparison, far less understood. 
Zhang and Zhang \cite{ZhangZhangSignificance}, and B\"uhlmann
\cite{BuhlmannSignificance} proposed hypothesis testing procedures
under restricted eigenvalue or compatibility conditions
\cite{buhlmann2011statistics}. These papers provide deterministic
guarantees but --in order to achieve a certain target significance
level $\alpha$ and
power $1-\beta$-- they require
$|\theta_{0,i}|\ge  c\,\max\{\sigma s_0 \log p/\, n,
\sigma/\sqrt{n}\}$. The best lower bound
\cite{javanmard2013hypothesis} shows that any such test 
requires instead $|\theta_{0,i}|\ge c(\alpha,\beta)\sigma/\sqrt{n}$. 
(The lower bound of \cite{javanmard2013hypothesis} is reproduced as
Theorem \ref{thm:GeneralUpperBound} here, for the reader's convenience.)

In other words, the guarantees of
\cite{ZhangZhangSignificance,BuhlmannSignificance}
can be suboptimal by a factor as large as  $\sqrt{s_0}$. Equivalently, in
order for the coefficient $\theta_{0,i}$ to be detectable with appreciable probability,
it needs to be larger than the overall $\ell_2$ error.
Here we will propose a test that --for random designs--
achieves  significance
level $\alpha$ and
power $1-\beta$ for $|\theta_{0,i}|\ge c'(\alpha,\beta)\sigma/\sqrt{n}$.

Lockhart et al. \cite{lockhart2013significance} develop a test for the
hypothesis that a newly added coefficient along the LASSO regularization path is
irrelevant.  This however does not allow to test arbitrary coefficients
at a given value of $\lambda$, which is instead the problem addressed
in this paper.
These authors further assume that the current LASSO support contains the actual 
support $\supp(\theta_0)$ and that
the latter has bounded size.

Belloni, Chernozhukov and  
collaborators~\cite{belloni2011inference,Belloni-Logistic} consider
inference in a regression model with high-dimensional data.
In this model the response variable relates to a scalar main regressor
and a $p$-dimensional control vector.
The main regressor is of primary interest and the control vector is treated as nuisance component.
Assuming that the control vector is $s_0$-sparse, the authors propose
a method to construct confidence regions for the parameter of interest
under the sample size requirement $(s_0^2 \log p)/n \to 0$.
The proposed method is shown to attain the semi-parametric efficiency bounds for this class of models.
The key modeling assumption in this paper is that the scalar regressor
of interest is random, and depends linearly on  the $p$-dimensional
control vector, with a sparse coefficient vector (with sparsity again
of order $o(\sqrt{n/\log p})$. This assumption is closely related to
the sparse inverse covariance assumption of \cite{GBR-hypothesis}
(with the difference that only one regressor is tested).

Finally, resampling methods for hypothesis testing were studied in 
\cite{MeinshausenBuhlmannStability,minnier2011perturbation}.
These methods are perturbation-based procedures to 
approximate the distribution
of a general class of penalized parameter estimates for the case $n> p$.
The idea is to consider the minimizer of a stochastically perturbed
version of the regularized objective function, call it $\tilde{\theta}$, and characterize the limiting distribution
of the regularized estimator $\htheta$ in terms of the distribution of $\tilde{\theta}$. In order to estimate the latter,
a large number of random samples of the perturbed objective function are generated,
and for each sample the minimizer is computed. Finally the theoretical distribution
of $\tilde{\theta}$ is approximated by the empirical distribution of these minimizers.

After the present paper was submitted for publication, we became aware
that B\"uhlmann and Dezeure \cite{BuhlmannUnpublished} had independently worked on
similar ideas.

\subsection{Preliminaries and notations}
In this section we introduce some basic definitions used
throughout the paper, starting with simple notations.

For a matrix $A$ and set of indices $I,J$, we let $A_{I,J}$ denote the submatrix formed 
by the rows in $I$ and columns in $J$. Also,
$A_{I,\cdot}$ (resp. $A_{\cdot,I}$) denotes the submatrix
containing just the rows (reps. columns) in $I$. 
Likewise, for a vector $v$, $v_I$ is the restriction
of $v$ to indices in $I$. 
We use the shorthand $A^{-1}_{I,J} = (A^{-1})_{I,J}$. In particular, $A^{-1}_{i,i} = (A^{-1})_{i,i}$.
The maximum and the minimum singular values of $A$ are respectively denoted 
by $\sigma_{\max}(A)$ and $\sigma_{\min}(A)$.
We write $\|v\|_p$ for the standard $\ell_p$ norm of a vector $v$, i.e., $\|v\|_p = (\sum_{i} |v_i|^p)^{1/p}$.
and  $\|v\|_0$ for the number of nonzero entries of  $v$. 
For a matrix $A$, $\|A\|_p$ is the $\ell_p$ operator norm, and $|A|_p$ is the elementwise $\ell_p$ norm.
For a vector $v$, $\supp(v)$ represents
the positions of nonzero entries of $v$. 
 Throughout, $\Phi(x) \equiv \int_{-\infty}^x e^{-t^2/2} \de t/\sqrt{2\pi}$ denotes the CDF of the 
standard normal distribution.
Finally, \emph{with high
probability} (w.h.p) means with probability converging to
one as $n \to \infty$.

We  let $\hSigma \equiv \bX^\sT\bX/n$ be the sample covariance matrix.
For  $p > n$, $\hSigma$ is always singular. However, we may require $\hSigma$ to be nonsingular for a restricted set
of directions. 
\begin{definition}
Given a symmetric matrix $\hSigma\in\reals^{p\times p}$ and a set
$S\subseteq [p]$, the corresponding \emph{compatibility constant}
is defined as
\begin{align}
\phi^2(\hSigma,S) \equiv
\min_{\theta\in\reals^p}\Big\{\frac{|S|\,\<\theta,\hSigma\,\theta\>}{\|\theta_S\|_1^2} :\;\;
\theta\in\reals^p, 
\;\; \|\theta_{S^c}\|_1\le 3\|\theta_S\|_1\Big\}\, .
\end{align}
We say that $\hSigma\in\reals^{p\times p}$ satisfies the
\emph{compatibility condition} for the  set $S\subseteq [p]$, with
constant $\phi_0$ if $\phi(\hSigma,S)\ge \phi_0$.
We say that it holds for the design matrix $\bX$, if it holds for
$\hSigma = \bX^{\sT}\bX/n$.
\end{definition}
In the following, we shall drop the argument $\hSigma$ if clear from
the context.
Note that a slightly more general definition is used normally
\cite[Section 6.13]{buhlmann2011statistics},
whereby  the condition $\|\theta_{S^c}\|_1 \le 3 \|\theta_S\|_1$, is
replaced by  $\|\theta_{S^c}\|_1 \le L \|\theta_S\|_1$. The resulting
constant $\phi(\hSigma,S,L)$ depends on $L$. For the sake of simplicity, we restrict
ourselves to the case $L=3$.

\begin{definition}
The \emph{sub-gaussian norm} of a random variable $X$, denoted by $\|X\|_{\psi_2}$, is defined as
\[
\|X\|_{\psi_2} = \sup_{q\ge 1}\, q^{-1/2} (\E |X|^q)^{1/q}\,.
\]
For a random vector $X \in \reals^n$, its sub-gaussian norm is defined as $\|X\|_{\psi_2} = \sup_{x\in S^{n-1}} \|\<X,x\>\|_{\psi_2}$,
where $S^{n-1}$ denotes the unit sphere in $\reals^n$.
\end{definition}

\begin{definition}
The \emph{sub-exponential norm} of a random variable $X$, denoted by $\|X\|_{\psi_1}$, is defined as
\[
\|X\|_{\psi_1} = \sup_{q\ge 1}\, q^{-1} (\E |X|^q)^{1/q}\,.
\]
For a random vector $X \in \reals^n$, its sub-exponential norm is defined as $\|X\|_{\psi_1} = \sup_{x\in S^{n-1}} \|\<X,x\>\|_{\psi_1}$, where $S^{n-1}$ denotes the unit sphere in $\reals^n$.
\end{definition}
%


\section{Compensating the bias of the LASSO}
\label{sec:DebiasedGen}

In this section we present our characterization of the de-biased
estimator
$\htheta^u$ (subsection \ref{sec:Debiased}). 
This characterization also clarifies in what sense the LASSO estimator
is biased. We discuss this point in subsection \ref{sec:BiasDiscussion}.

\subsection{A de-biased estimator for $\theta_0$}
\label{sec:Debiased}

As emphasized above, our approach is based on a de-biased estimator
defined in Eq.~\eqref{eq:hthetau},
and on its distributional properties. In order to clarify the latter,
it is convenient to begin with a slightly broader setting and 
consider a general debiasing procedure that makes use of a an arbitrary
$M\in\reals^{p\times p}$. Namely, we define
\begin{align}
\htheta^*(Y,\bX;M,\lambda)= \htheta^n(\lambda) + \frac{1}{n}\, M
\bX^\sT(Y - \bX \htheta^n(\lambda)) \, .
\label{eq:GeneralDebiased}
\end{align}
For notational simplicity, we shall omit the arguments
$Y,\bX,M,\lambda$ unless they are required for clarity.
The quality of this debiasing procedure depends of course on the choice of $M$,
as well as on the design $\bX$. We characterize the pair $(\bX,M)$  by
the following figure of merit.
\begin{definition}\label{def:Coherence}
Given the pair $\bX\in\reals^{n\times p}$ and $M\in \reals^{p\times
  p}$, let $\hSigma = \bX^{\sT}\bX/n$ denote the associated sample
covariance. Then, the \emph{generalized coherence parameter of $\bX,M$},
denoted by $\coh_*(\bX;M)$, is
\begin{align}
\coh_*(\bX;M) \equiv \big|M\hSigma-\id\big|_{\infty}\, .
\end{align}
The \emph{minimum (generalized) coherence} of $\bX$ is $\com(\bX) =
\min_{M\in\reals^{p\times p}} \coh_*(\bX;M)$. We denote by $M_{\rm
  min}(\bX)$ any minimizer of $\coh_*(\bX;M)$.
\end{definition}
Note that the minimum coherence can be computed efficiently since
$M\mapsto \coh_*(\bX;M)$ is a convex function (even more, the
optimization problem is a linear program).

The motivation for our terminology can be grasped by considering the
following special case.
\begin{remark}
Assume that the columns of $\bX$ are normalized to have
$\ell_2$ norm equal to $\sqrt{n}$ (i.e. $\|\bX e_i\|_2 = \sqrt{n}$ for
all $i\in [p]$), and $M= \id$. Then $(M\hSigma-\id)_{i,i} = 0$, and the maximum
$|M\hSigma-\id|_{\infty} = \max_{i\neq j}|(\hSigma)_{ij}|$. In other words
$\coh(\bX;\id)$ is the maximum normalized scalar product between
distinct columns of $\bX$:
\begin{align}
\coh_*(\bX;\id) = \frac{1}{n}\max_{i\neq j} \big|\<\bX e_i, \bX e_j\>\big|\, . \label{eq:StdCoherence}
\end{align}
\end{remark}
The quantity (\ref{eq:StdCoherence}) is known as the \emph{coherence
  parameter} of the matrix $\bX/\sqrt{n}$ and was first defined in the context
of approximation theory by Mallat and Zhang \cite{mallat1993matching}, and  by Donoho
and Huo \cite{donoho2001uncertainty}.

Assuming, for the sake of
simplicity, that the columns of $\bX$ are normalized so that 
$\|\bX e_i\|_2 = \sqrt{n}$, a small value of the coherence parameter $\coh_*(\bX;\id)$ means
that the columns of $\bX$ are roughly orthogonal. 
We emphasize however that $\coh_*(\bX;M)$ can be much smaller than its
classical coherence parameter $\coh_*(\bX;\id)$. For instance, $\coh_*(\bX;\id)=0$ if and only if $\bX/\sqrt{n}$
is an orthogonal matrix. On the other hand, $\com(\bX) = 0$ if and only if $\bX$ has rank\footnote{Of
course this example requires $n\ge p$. It is the
simplest example that illustrates the difference between coherence and
generalized coherence, and it is not hard to find related examples
with $n<p$.} $p$.

The following theorem is a slight generalization of a result of
\cite{GBR-hypothesis}. Let us emphasize that it applies to deterministic
design matrices $\bX$.
\begin{thm}\label{thm:deterministic}
Let $\bX\in \reals^{n\times p}$ be any (deterministic)
design matrix, and $\htheta^*= \htheta^*(Y,\bX;M,\lambda)$ be a
general debiased 
estimator as per Eq.~(\ref{eq:GeneralDebiased}).
Then, setting $Z = M\bX^{\sT}W/\sqrt{n}$, we have
\begin{align}
\sqrt{n} (\htheta^* - \theta_0) = Z + \Delta\,, \quad 
Z  \sim \normal(0,\sigma^2 M\hSigma M^\sT)\,, \quad
\Delta = \sqrt{n} (M\hSigma-\id) (\theta_0 - \htheta^n)\,. \label{eq:GeneralRepresentation}
\end{align}
Further, assume that $\bX$ satisfies the compatibility condition for
the set $S=\supp(\theta_0)$, $|S|\le s_0$, with constant $\phi_0$, and has generalized coherence parameter 
$\mu_*=\mu_*(\bX;M)$, and let $K\equiv
\max_{i\in[p]} (\bX^{\sT}\bX/n)_{ii}$. Then, letting $\lambda = \sigma\sqrt{(c^2\log
p)/n}$, we have
\begin{align}
\prob\Big(\|\Delta\|_{\infty} \ge\frac{4 c \mu_*\sigma s_0}{\phi_0^2} 
\sqrt{\log p}  \Big) \le 2p^{-c_0} \,,\;\;\;\; c_0 = \frac{c^2}{32K} -1\, . \label{eq:FirstBoundDelta}
\end{align}
Further, if $M = M_{\rm min}(\bX)$ minimizes the convex cost function
$|M\hSigma-\id|_{\infty}$, then $\mu_*$ can be replaced by $\com(\bX)$
in Eq.~(\ref{eq:FirstBoundDelta}).
\end{thm}

The above theorem decomposes the estimation error $(\htheta^*-\theta_0)$ 
into a zero mean Gaussian term $Z/\sqrt{n}$ and a bias term $\Delta/\sqrt{n}$
whose maximum entry is bounded as per Eq.~(\ref{eq:FirstBoundDelta}). 
This estimate on $\|\Delta\|_{\infty}$ depends on the design matrix through two constants: the 
compatibility constant $\phi_0$ and the generalized coherence
parameter $\mu_*(\bX;M)$. The former is a well studied property of the
design matrix
\cite{buhlmann2011statistics,BuhlmannVanDeGeer}, 
and assuming $\phi_0$ of order one is nearly necessary for the LASSO to achieve optimal estimation
rate in high dimension. On the contrary, the definition of
$\mu_*(\bX;M)$ is a new contribution of the present paper.

The next theorem establishes that, for a natural probabilistic model
of the design matrix $\bX$, both $\phi_0$ and $\mu_*(\bX;M)$ can
be bounded with probability converging rapidly to one as
$n,p\to\infty$. Further, the bound  on $\mu_*(\bX,M)$ hold for the
special choice of $M$ that is constructed by Algorithm 1.
\begin{thm}\label{thm:event_thm}
Let $\Sigma\in\reals^{p\times p}$ 
be such that $\sigma_{\min}(\Sigma) \ge C_{\min} > 0$, and
$\sigma_{\max}(\Sigma) \le C_{\max} < \infty$,
and $\max_{i\in [p]}\Sigma_{ii}\le 1$.
Assume $\bX\Sigma^{-1/2}$ to have independent subgaussian  rows, with
zero mean and subgaussian norm $\|\Sigma^{-1/2} X_1\|_{\psi_2} = \kappa$, for some
constant $ \kappa \in(0, \infty)$.
\begin{enumerate}
\item[$(a)$] For $\phi_0,s_0,K\in\reals_{>0}$, let 
$\event_n=\event_n(\phi_0,s_0,K)$ be the event that the compatibility
condition holds for $\hSigma=(\bX^{\sT}\bX/n)$,  
for all sets $S\subseteq [p]$, $|S|\le s_0$  with constant $\phi_0>0$, and
that $\max_{i\in [p]}\, \hSigma_{i,i} \le K$. Explicitly 
\begin{align}
\event_n(\phi_0,s_0,K) \equiv\Big\{\bX\in\reals^{n\times
  p}:\,\;\min_{S:\; |S|\le s_0}\phi(\hSigma,S)\ge \phi_0\, ,
\max_{i\in [p]}\, \hSigma_{i,i} \le K, \;\; \hSigma = (\bX^{\sT}\bX/n)\Big\}\, .
\end{align}
Then there exists $c_*\le 2000$ such that the following happens. If
$n\ge\nu_0\,  s_0\log (p/s_0)$, $\nu_0 \equiv 4c_*(C_{\rm
  max}\kappa^4/C_{\rm min})$,  $\phi_0 = C_{\rm min}^{1/2}/2$, and $K\ge
1+20\kappa^2\sqrt{(\log p)/n}$, then
\begin{align}
\prob\big(\bX\in \cE_n(\phi_0,s_0,K) \big) \ge 1-4\, e^{-c_1n}\,
,\;\;\;\;\;\;\; c_1\equiv  \frac{1}{c_*\kappa^4}\, . \label{eq:BoundEvA}
\end{align}
\item[$(b)$]  For $a>0$,  $\coev_n=\coev_n(a)$ be the event that the problem
  (\ref{eq:optimization}) is feasible for $\coh = a\sqrt{(\log p)
    /n}$, or equivalently
\begin{align}
\coev_n(a) \equiv \Big\{\bX\in\reals^{n\times p}:\; \com(\bX) <
a\sqrt{\frac{\log p}{n}}\Big\}\, .
\end{align}
Then, for $n\ge  a^2C_{\min}\log p/(4e^2C_{\max}\kappa^4) $
\begin{align}
\prob\big(\bX\in\coev_n(a)\big) \ge 1-2\, p^{-c_2}\, ,\;\;\;\;\;\;\;
c_2\equiv  \frac{a^2 C_{\min}}{24e^2\kappa^4 C_{\max}} - 2\,.
\label{eq:BoundEvG}
\end{align}
\end{enumerate}
\end{thm}
The proof of  this theorem
is given in Section~\ref{proof:thm_eventA} (for part $(a)$) and Section
\ref{proof:thm_eventB} (part $(b)$).

The proof that event $\event_n$ holds with high probability relies
crucially on a theorem by Rudelson and Zhou \cite[Theorem 6]{rudelson2011reconstruction}. 
Simplifying somewhat, the latter states that, if the restricted
eigenvalue condition of \cite{BickelEtAl} holds for the population
covariance $\Sigma$, then it holds with high probability for the sample
covariance $\hSigma$. (Recall that the restricted eigenvalue condition is
implied by a lower bound on the minimum singular value\footnote{Note, in
particular, at the cost of further complicating the last statement, 
the condition $\sigma_{\rm min}(\Sigma)=\Omega(1)$ can be further weakened.}, and that it
implies the compatibility condition \cite{BuhlmannVanDeGeer}.)

Finally, by putting together Theorem \ref{thm:deterministic} and Theorem
\ref{thm:event_thm},
we obtain the following conclusion.
\begin{thm}\label{thm:main_thm}
Consider the linear model~\eqref{eqn:regression} and let $\htheta^u$ be defined as per
Eq.~\eqref{eq:hthetau} in Algorithm 1, with $\mu =a\sqrt{(\log p)/n}$.
Then,  setting $Z = M\bX^{\sT}W/\sqrt{n}$, we have
\begin{align}
\sqrt{n} (\htheta^u - \theta_0) = Z + \Delta\,, \quad 
Z | \bX \sim \normal(0,\sigma^2 M\hSigma M^\sT)\,, \quad
\Delta = \sqrt{n} (M\hSigma-\id) (\theta_0 - \htheta^n)\,. \label{eq:GeneralRepresentationBis}
\end{align}
Further, under the assumptions of Theorem \ref{thm:event_thm}, 
and for $n\ge \max(\nu_0s_0\log (p/s_0), \nu_1\log p)$, 
$\nu_1 = \max(1600\kappa^4, a/4)$,
and 
$\lambda = \sigma\sqrt{(c^2\log
p)/n}$, we have
\begin{align}
\prob\left\{\|\Delta\|_{\infty} \ge
\Big(\frac{16ac\, \sigma}{C_{\rm min}} \Big)\frac{s_0\log p}{\sqrt{n}} \right\}\le  4\,e^{-c_1n}  + 4\,p^{-\tilde{c_0}\wedge c_2}\,.
\label{eq:TailBoundQuant}
\end{align}
where $\tilde{c_0} = (c^2/48)-1$ and $c_1, c_2$ are given by
Eqs.~(\ref{eq:BoundEvA}) and (\ref{eq:BoundEvG}). 

Finally, the tail bound
(\ref{eq:TailBoundQuant}) holds for any choice of $M$ that is only
function of the design matrix $\bX$, and satisfies the feasibility
condition in Eq.~(\ref{eq:optimization}),
i.e. $|M\hSigma-\id|_{\infty}\le \coh$.
\end{thm}
Assuming $\sigma,C_{\rm min}$ of order one,
the last theorem establishes that, for random designs, the maximum
size of the `bias term' $\Delta_i$ over $i\in[p]$ is:
\begin{align}
\|\Delta\|_{\infty} = O\Big(\frac{s_0\log p}{\sqrt{n}}\Big)
\end{align}
On the other hand, the `noise term' $Z_i$ is roughly of order $\sqrt{[M\hSigma M^\sT]_{ii}}$.
Bounds on the variances $[M\hSigma M^\sT]_{ii}$ will be
given in Section \ref{sec:HypothesisTesting}  showing that, if 
$M$ is computed through Algorithm 1, $[M\hSigma M^\sT]_{ii}$ is of
order one for a broad family of random designs. 
As a consequence $|\Delta_i|$ is much smaller than $|Z_i|$ 
whenever $s_0 = o(\sqrt{n}/\log p)$. 
We summarize these remarks below.
\begin{remark}\label{rem:sparsity_cond}
Theorem~\ref{thm:main_thm} only requires that the support size
satisfies $s_0= O(n/\log p)$.
If we further assume $s_0 = o(\sqrt{n}/\log p)$, then we have
$\|\Delta\|_\infty = o(1)$
with high probability.
Hence, $\htheta^u$ is an asymptotically unbiased estimator for
$\theta_0$. 
\end{remark}
A more formal comparison of the bias of $\htheta^u$, and of the one of
the LASSO estimator $\htheta^n$ can be found in Section \ref{sec:BiasDiscussion} below.
Section \ref{sec:GBR-comparison} compares our approach with the related one in \cite{GBR-hypothesis}.

As it can be seen from the statement of Theorem
\ref{thm:deterministic} and Theorem \ref{thm:event_thm},
the claim of Theorem~\ref{thm:main_thm} does not rely on the specific choice of the objective function
in optimization problem~\eqref{eq:optimization} and only uses the
constraint on $\|\hSigma m-e_i\|_{\infty}$. In particular it holds for any
matrix $M$ that is feasible.  On the other hand, the specific
objective function problem~\eqref{eq:optimization} minimizes the
variance of the noise term $\Var(Z_i)$.
%
%
\subsection{Discussion: The bias of the LASSO}
\label{sec:BiasDiscussion}

Theorems \ref{thm:deterministic} and 
\ref{thm:event_thm} provide a quantitative framework  to discuss  in what
sense  the LASSO estimator $\htheta^n$ is asymptotically biased, while the de-biased
estimator $\htheta^u$ is asymptotically unbiased. 

Given an estimator $\htheta^n$ of the parameter vector $\theta_0$, we define its \emph{bias} to be the
vector 
\begin{align}
\bias(\htheta^n) \equiv \E\{\htheta^n-\theta_0|\bX\}\, .
\end{align}
Note that, if the design is random, $\bias(\htheta^n)$ is a measurable
function of $\bX$. If the design is deterministic, $\bias(\htheta^n)$ is
a deterministic quantity as well, and the conditioning is redundant.

It follows from Eq.~(\ref{eq:GeneralRepresentation}) that
\begin{align}
\bias(\htheta^u) = \frac{1}{\sqrt{n}}\E\{\Delta|\bX\}\, .
\end{align}
Theorem  \ref{thm:main_thm}  with high probability,
 $\|\Delta\|_{\infty}= O(s_0\log p/\sqrt{n})$. The next
corollary establishes that this translates into a bound  on
$\bias(\htheta^u)$ for all $\bX$ in a set that has probability
rapidly converging to one as $n$, $p$ get large.
\begin{coro}\label{coro:UIsNotBiased}
Under the assumptions of Theorem  \ref{thm:main_thm}, 
let $c_1$, $c_2$ be defined as per Eqs.~(\ref{eq:BoundEvA}), (\ref{eq:BoundEvG}). Then we have
\begin{align}
\bX\in \event_n(C_{\min}^{1/2}/2,s_0,3/2)\cap \coev_n(a)
\; \Rightarrow\; \|\bias(\htheta^u)\|_{\infty}\le \frac{160 a}{C_{\min}}\, \frac{\sigma
  s_0\log p}{n}\, ,\label{eq:BoundCoroPerX}\\
\prob\Big(\bX\in \event_n(C_{\min}^{1/2}/2,s_0,3/2)\cap \coev_n(a)\Big)
\ge 1-4e^{-n/c_*}-2\,p^{-c_2}\, .\label{eq:CoroProbabilityBound}
\end{align}
\end{coro} 
The proof of this corollary can be found in Appendix
\ref{app:UIsNotBiased}.

This result can be contrasted with a converse  result for the LASSO
estimator. Namely, as stated below, there are choices of the vector
$\theta_0$, and of the design covariance $\Sigma$, such that
$\bias(\htheta^n)$ is the sum of two terms. One is of order 
order $\lambda = c\sigma\sqrt{(\log p)/n}$ and the second is of order
$\|\bias(\htheta^u)\|_{\infty}$.
If $s_0$ is significantly smaller than $\sqrt{n/\log p}$ (which is the
main regime studied in the rest of the paper), the first term
dominates and $\|\bias(\htheta^n)\|_{\infty}$ is much larger than
$\|\bias(\htheta^u)\|_{\infty}$.
If on the other hand $s_0$ is significantly larger than $\sqrt{n/\log
  p}$ 
then $\|\bias(\htheta^n)\|_{\infty}$ is of the same order as $\|\bias(\htheta^u)\|_{\infty}$.
This justify referring to $\htheta^u$ as to an \emph{unbiased
  estimator}.

Notice that, since we want to establish a negative result about the
LASSO, it is sufficient to exhibit a specific covariance structure $\Sigma$ satisfying
the assumptions of the previous corollary. Remarkably it is sufficient
to consider standard designs, i.e. $\Sigma = \id_{p\times p}$.
\begin{coro}\label{coro:LASSOIsBiased}
Under the assumptions of Theorem  \ref{thm:main_thm}, further consider
the case $\Sigma = \id$. 
Then, there exists a
numerical constant  $c_{**}>0$, a set of design matrices
$\cB_n\subseteq \reals^{n\times p}$, and coefficient vectors
$\theta_0\in \reals^p$, $\|\theta_0\|_0\le s_0$, such that
\begin{align}
&\bX\in\cB_n\;\Rightarrow\;
\|\bias(\htheta^n)\|_{\infty}\ge  \left|\frac{2}{3}\lambda-
\|\bias(\htheta^u)\|_{\infty}\right|\, ,\label{eq:CoroBiasStd}\\
&\prob(\cB_n) \ge 1-6\, e^{-n/c_*}-2\, p^{-3} \, .\label{eq:CoroPbound2}\\
\end{align}
In particular $\|\bias(\htheta^u)\|_{\infty}\le \lambda/3$ (which
follows from $(s_0^2\log p)/n\le  (c/(3c_{**}))^2$) then  we have
\begin{align}
\|\bias(\htheta^n)\|_{\infty}\ge  \frac{c\sigma}{3}\sqrt{\frac{\log
    p}{n}}\gg \|\bias(\htheta^u)\|_{\infty}\, . \label{eq:CoroBiasStd_final}
\end{align}
On the other hand, if 
$\|\bias(\htheta^u)\|_{\infty}\ge \lambda$, then
\begin{align}
\|\bias(\htheta^n)\|_{\infty}\ge  \frac{1}{3}\|\bias(\htheta^u)\|_{\infty}\, . \label{eq:CoroBiasStd_final2}
\end{align}
\end{coro}
A formal proof of this statement is deferred to Appendix
\ref{app:LASSOIsBiased}, but  the underlying mathematical mechanism is
quite simple and instructive.
Recall that the KKT conditions for the LASSO estimator
(\ref{eq:LASSOEstimator}) read
\begin{align}
\frac{1}{n}\bX^{\sT}(Y-\bX\htheta^n) = \lambda\, v(\htheta^n)\, ,
\end{align}
with $v(\htheta^n)\in\reals^p$ a vector in the subgradient of the $\ell_1$ norm
at $\htheta^n$. Adding $\htheta^n-\theta_0$ to both sides, and taking
expectation over the noise, we get
\begin{align}
\bias(\htheta^*) = \bias(\htheta^n) + \lambda\E\{v(\htheta^n)|\bX\}\, ,\label{eq:TwoBiasRelation}
\end{align}
Where  $\htheta^*$ a debiased estimator of the general form
Eq.~(\ref{eq:GeneralDebiased}), for $M=\id$.
This suggest that $\bias(\htheta^n)$ can be decomposed in two
contributions as described above, and as shown formally in Appendix \ref{app:LASSOIsBiased},

\subsection{Comparison with earlier results}
\label{sec:GBR-comparison}

In this Section we briefly compare the above debiasing procedure
and in particular Theorems \ref{thm:deterministic},
\ref{thm:event_thm} and \ref{thm:main_thm} to the results of \cite{GBR-hypothesis}.
In the case of linear statistical models considered here, 
the authors of \cite{GBR-hypothesis} construct  a debiased estimator
of the form (\ref{eq:GeneralDebiased}). However, instead of solving
the optimization problem (\ref{eq:optimization}), they follow
\cite{ZhangZhangSignificance} and use the regression coefficients of the $i$-th
column of $\bX$ on the other columns to construct the $i$-th row of
$M$.  These regression coefficients are computed --once again-- using
the LASSO (node-wise LASSO).

It useful to spell out the most important differences between our
contribution and the ones of \cite{GBR-hypothesis}:
\begin{enumerate}
\item The case of fixed non-random designs is covered by \cite[Theorem
  2.1]{GBR-hypothesis}, which should be compared to our Theorem
  \ref{thm:deterministic}.
  While in our case the bias is controlled by the generalized
  coherence parameter, a similar role is played in
  \cite{GBR-hypothesis} by the regularization parameters of the
  nodewise LASSO.
\item The case of random designs is covered by \cite[Theorem
  2.2, Theorem 2.4]{GBR-hypothesis}, which should be compared with our Theorem
  \ref{thm:main_thm}.
In this case, the assumptions underlying our result are significantly less
restrictive. More precisely:
\begin{enumerate}
\item \cite[Theorem
  2.2, Theorem  2.4]{GBR-hypothesis} assume $\bX$ to have i.i.d. rows, while we
  only assume the rows to be independent.
\item  \cite[Theorem
  2.2, Theorem  2.4]{GBR-hypothesis} assume the rows inverse covariance matrix
  $\Sigma^{-1}$ be sparse. More precisely, letting $s_j$ be the number
  of non-zero entries of the $j$-th  row of $\Sigma^{-1}$,
  \cite{GBR-hypothesis} assumes $\max_{j\in [p]}s_j = o(n/\log p)$,
  that is much smaller than $p$. 
We do not make any sparsity assumption for $\Sigma^{-1}$, and $s_j$ can be as large as $p$.
\end{enumerate}
(In fact \cite[Theorem  2.4]{GBR-hypothesis} also consider the
assumption of $\bX$ with
bounded entries, but even stricter sparsity assumptions are made in
that case.)
\end{enumerate}
In addition our Theorem \ref{thm:main_thm} provides the specific
dependence on the maximum and minimum singular value of $\hSigma$.

Let us also note that solving the convex problem
(\ref{eq:optimization}) is not more burdensome than solving the
nodewise LASSO as in \cite{ZhangZhangSignificance,GBR-hypothesis},
This can be confirmed by checking that the dual of the problem (\ref{eq:optimization})
is an $\ell_1$-regularized quadratic optimization problem. It has
therefore the same complexity as the nodewise LASSO (but it is
different from the nodewise LASSO).

%
\section{Statistical inference}
\label{sec:Inference}

A direct application of Theorem~\ref{thm:main_thm} is to derive confidence intervals
and statistical hypothesis tests for high-dimensional models. 
Throughout, we make the sparsity assumption $s_0 = o(\sqrt{n}/\log p)$
and omit explicit constants that can be readily derived from Theorem
\ref{thm:main_thm}.

\subsection{Preliminary lemmas}

As discussed above, the bias term $\Delta$ is negligible with respect
to the random term $Z$ in the decomposition
(\ref{eq:GeneralRepresentationBis}),
provided the latter has variance of order one.
Our first lemma establishes that this is indeed the case. 
\begin{lemma}\label{lem:missing_bound}
Let $M = (m_1,\dotsc,m_p)^\sT$ be the matrix with rows $m_i^\sT$ obtained by solving convex
program~\eqref{eq:optimization} in Algorithm 1. Then for all $i\in [p]$,
\[[M \hSigma M^\sT]_{i,i} \ge \frac{(1-\coh)^2}{\hSigma_{i,i}}\,.\]
\end{lemma} 
Lemma~\ref{lem:missing_bound} is proved in Appendix~\ref{app:missing_bound}.

Using this fact, we can then characterize the asymptotic distribution
of the residuals $(\htheta^u-\theta_{0,i})$. Theorem
\ref{thm:main_thm} naturally suggests to consider the scaled residual
$\sqrt{n}(\htheta^u_i - \theta_{0,i})/(\sigma [M \hSigma
M^\sT]_{i,i}^{1/2})$. In the next lemma we consider a slightly more
general scaling,  replacing $\sigma$ by a consistent estimator
$\hsigma$.
\begin{lemma}\label{lemma:LastDistribution}
Consider a sequence of design matrices $\bX\in\reals^{n\times p}$,
 with dimensions $n\to\infty$, $p=p(n)\to\infty$ satisfying the following assumptions,
for constants $C_{\rm min}, C_{\rm max},\kappa\in (0,\infty)$
independent of $n$.
For each $n$, $\Sigma\in\reals^{p\times p}$  is such that $\sigma_{\min}(\Sigma) \ge C_{\min} > 0$, and
$\sigma_{\max}(\Sigma) \le C_{\max} < \infty$,
and $\max_{i\in [p]}\Sigma_{ii}\le 1$.
Assume $\bX\Sigma^{-1/2}$ to have independent subgaussian  rows, with
zero mean and subgaussian norm $\|\Sigma^{-1/2} X_1\|_{\psi_2} \le  \kappa$,

Consider the linear model~\eqref{eqn:regression} and let $\htheta^u$ be defined as per
Eq.~\eqref{eq:hthetau} in Algorithm 1, with $\coh =a\sqrt{(\log p)/n}$
and $\lambda = \sigma\sqrt{(c^2\log p)/n}$, with $a,c$  large enough
constants.
Finally, let $\hsigma = \hsigma(y,\bX)$ an estimator of the noise
level satisfying, for any $\eps>0$,
\begin{align}
\lim_{n\to\infty} \sup_{\theta_0\in\reals^p;\, \|\theta_0\|_0 \le s_0 }\prob\Big(\Big|\frac{\hsigma}{\sigma}-1\Big|\ge \eps
\Big)=0\, .\label{eq:ConsistencySigma}
\end{align}
If $s_0=o(\sqrt{n}/\log p)$ ($s_0\ge 1$), then, for all $x\in\reals$, we have
\begin{eqnarray}\label{eq:distribution}
\lim_{n\to\infty}\sup_{\theta_0\in\reals^p;\, \|\theta_0\|_0 \le s_0 }\left|\prob \left\{\frac{\sqrt{n}(\htheta^u_i - \theta_{0,i})}{\hsigma [M \hSigma M^\sT]_{i,i}^{1/2}} 
\le x  \right\} -\Phi(x)\right| =0\, . 
\end{eqnarray}
\end{lemma}
The proof of this lemma can be found in Section \ref{proof:LastDistribution}.
We also note that the dependence of $a,c$ on $C_{\rm min}, C_{\rm
  max}, \kappa$ can be easily reconstructed from Theorem \ref{thm:event_thm}.

The last lemma requires a consistent estimator of $\sigma$, in the
sense of Eq.~(\ref{eq:ConsistencySigma}). Several proposal have been
made to estimate the noise level in high-dimensional linear
regression. A short list of references includes
\cite{SCAD01,sis08,SBvdG10,mcp10,SZ-scaledLASSO,
BelloniChern,FGH12,reid2013study,dicker2012residual,isis13,bayati2013estimating}.
Consistency results have been proved or can be proved for several of
these estimators.

In order to demonstrate that the consistency criterion
(\ref{eq:ConsistencySigma})
can be achieved, we  use the scaled LASSO~\cite{SZ-scaledLASSO} given by
\begin{align}
\{\htheta^n(\tlambda), \hsigma(\tlambda)\} \equiv \underset{\theta\in\reals^p,\sigma> 0}{\arg\min}\,
\Big\{\frac{1}{2\sigma n}\|Y-\bX\theta\|^2_2+ \frac{\sigma}{2}
+\tlambda\|\theta\|_1\Big\}\, .
\label{eq:SLASSO}
\end{align}
This is a joint convex optimization problem which provides an estimate of the noise level in addition to an estimate of $\theta_0$.

The following lemma uses the analysis of \cite{SZ-scaledLASSO}  to
show that  $\hsigma$ thus defined satisfies  the consistency criterion
(\ref{eq:ConsistencySigma}).
\begin{lemma}\label{lemma:ConsistencySigma}
Under the assumptions of Lemma \ref{lemma:LastDistribution}, let 
$\hsigma=\hsigma(\tlambda)$ be the scaled LASSO estimator of the noise level, see
Eq.~(\ref{eq:SLASSO}), with $\tlambda = 10\sqrt{(2\log p)/n}$. Then
 $\hsigma$ thus satisfies Eq.~(\ref{eq:ConsistencySigma}).
\end{lemma}
The proof of this lemma is fairly straightforward and can be found in
Appendix \ref{sec:ConsistencySigma}.

\subsection{Confidence intervals}
\label{sec:ConfidenceInterval}

In view of Lemma \ref{lemma:LastDistribution}, it is quite
straightforward to construct asymptotically valid confidence
intervals. Namely, for $i\in [p]$ and significance level $\alpha\in
(0,1)$, we let
\begin{eqnarray}\label{eq:CI}
\begin{split}
J_i(\alpha) &
\equiv  [\htheta^u_{i} -\delta(\alpha,n), \htheta^u_{i}
+\delta(\alpha,n)]\, ,\\
 \delta(\alpha,n) &\equiv \Phi^{-1}(1-\alpha/2) \,
 \frac{\hsigma}{\sqrt{n}}
 [M   \hSigma M^\sT]^{1/2}_{i,i}\,.
\end{split}
\end{eqnarray}
\begin{thm}\label{coro:Interval}
Consider a sequence of design matrices $\bX\in\reals^{n\times p}$,
 with dimensions $n\to\infty$, $p=p(n)\to\infty$ satisfying the
 assumptions of Lemma \ref{lemma:LastDistribution}.

Consider the linear model~\eqref{eqn:regression} and let $\htheta^u$ be defined as per
Eq.~\eqref{eq:hthetau} in Algorithm 1, with $\coh =a\sqrt{(\log p)/n}$
and $\lambda = \sigma\sqrt{(c^2\log p)/n}$, with $a,c$  large enough constants.
Finally, let $\hsigma = \hsigma(y,\bX)$ a consistent estimator of the noise
level in the sense of Eq.~(\ref{eq:ConsistencySigma}). Then the
confidence interval $J_i(\alpha)$ is asymptotically valid, namely
\begin{align}
\lim_{n\to\infty}\prob\Big(\theta_{0,i}\in J_{i}(\alpha)\Big) =
1-\alpha\, .
 \end{align}
\end{thm}
\begin{proof}
The proof is an immediate consequence of Lemma
\ref{lemma:LastDistribution} since
\begin{align}
\lim_{n\to\infty}\prob\Big(\theta_{0,i}\in J_{i}(\alpha)\Big) =&
\lim_{n\to\infty}\prob \left\{\frac{\sqrt{n}(\htheta^u_i - \theta_{0,i})}{\hsigma [M \hSigma M^\sT]_{i,i}^{1/2}} 
\le \Phi^{-1}(1-\alpha/2) \right\} \\
&-
\lim_{n\to\infty}\prob \left\{\frac{\sqrt{n}(\htheta^u_i - \theta_{0,i})}{\hsigma [M \hSigma M^\sT]_{i,i}^{1/2}} 
\le -\Phi^{-1}(1-\alpha/2) \right\}\\
 = &1-\alpha
\, .
 \end{align}
\end{proof}

\subsection{Hypothesis testing}
\label{sec:HypothesisTesting}

An important advantage of sparse linear regression models is that they
provide parsimonious explanations of the data in terms of a small
number of covariates. The easiest way to select the `active'
covariates is to choose the indexes $i$ for which $\htheta_i^n\neq
0$. This approach however does not provide a measure of statistical
significance for the finding that the coefficient is non-zero.

More precisely, we are interested in testing an
individual null hypothesis $H_{0,i} : \theta_{0,i} = 0$ versus the alternative $H_{A,i}: \theta_{0,i} \neq 0$, and
assigning $p$-values for these tests.
We construct a $p$-value $P_i$ for the test $H_{0,i}$ as follows:
\begin{eqnarray}
P_i = 2\bigg(1 - \Phi\bigg(\frac{\sqrt{n}\, |\htheta^u_i|}{\hsigma [M \hSigma M^\sT]_{i,i}^{1/2}} \bigg)\bigg)\,.  \label{eq:p-value}
\end{eqnarray}
The decision rule is then based on the $p$-value $P_i$:
\begin{eqnarray}
\begin{split}\label{eq:decision-rule}
\hT_{i,\bX}(y) = \begin{cases}
1 & \text{if } P_i \le \alpha \quad \quad \text{ (reject $H_{0,i}$)}\,,\\
0 & \text{otherwise} \quad\quad \text{(accept $H_{0,i}$)}\,,
\end{cases}
\end{split}
\end{eqnarray} 
where $\alpha$ is the fixed target Type I error probability.
We measure the quality of the test $\hT_{i,\bX}(y)$ in terms of its
significance level $\alpha_i$ 
 and statistical power $1-\beta_i$. Here $\alpha_i$ is the probability
 of type I error (i.e. of a false positive at $i$)
 and $\beta_i$ is the probability of type II error (i.e. of a false
 negative at $i$).

Note that it is important to consider the tradeoff between
statistical significance and power. Indeed any significance level
$\alpha$ can be achieved  by randomly rejecting $H_{0,i}$ with
probability $\alpha$. This test achieves power
$1-\beta=\alpha$. Further note that,  without further assumption, no nontrivial power can be achieved. 
 In fact, choosing $\theta_{0,i} \neq 0$ arbitrarily close to zero, $H_{0,i}$ becomes indistinguishable from its alternative.
 We will therefore assume that, whenever $\theta_{0,i} \neq 0$, we have $|\theta_{0,i}| > \lb$ as well. 
 We take a minimax perspective and require the test to 
behave uniformly well over $s_0$-sparse vectors. Formally, given 
a family of tests $T_{i,\bX}:\reals^n\to\{0,1\}$, indexed by $i\in
[p]$, $\bX\in\reals^{n\times p}$, we define, for $\lb>0$ a lower bound on the
non-zero entries:
\begin{align}
\alpha_{i,n}(T) &\equiv \sup \Big\{\prob_{\theta_0}(T_{i,\bX}(y) = 1): \,\,  \theta_0 \in \reals^p,\,\, \|\theta_0\|_0\le s_0(n),\,\, \theta_{0,i} = 0 \Big\}\,. \label{eq:AlphaDef}\\
\beta_{i,n}(T; \lb) &\equiv \sup \Big\{\prob_{\theta_0}(T_{i,\bX}(y) = 0): \,\, \theta_0 \in \reals^p,\,\, \|\theta_0\|_0\le s_0(n),\,\, |\theta_{0,i}| \ge \lb \Big\}\,. \label{eq:BetaDef}
\end{align}
 Here, we made dependence on $n$ explicit. Also, $\prob_\theta(\,
 \cdot\, )$ denotes the induced probability for random design $\bX$ and noise
 realization $w$, given the fixed parameter vector $\theta$.
 Our next theorem establishes bounds on $\alpha_{i,n}(\hT)$ and
 $\beta_{i,n}(\hT;\lb)$ for our decision rule  (\ref{eq:decision-rule}).
\begin{thm}\label{thm:error-power}
Consider a sequence of design matrices $\bX\in\reals^{n\times p}$,
 with dimensions $n\to\infty$, $p=p(n)\to\infty$ satisfying the
 assumptions of Lemma \ref{lemma:LastDistribution}.

Consider the linear model~\eqref{eqn:regression} and let $\htheta^u$ be defined as per
Eq.~\eqref{eq:hthetau} in Algorithm 1, with $\coh =a\sqrt{(\log p)/n}$
and $\lambda = \sigma\sqrt{(c^2\log p)/n}$, with $a,c$  large enough constants.
Finally, let $\hsigma = \hsigma(y,\bX)$ a consistent estimator of the noise
level in the sense of Eq.~(\ref{eq:ConsistencySigma}),
and $\hT$ be the test defined in Eq.~(\ref{eq:decision-rule}).

Then the following holds true for any fixed sequence of
integers $i=i(n)$:
\begin{align}
\lim_{n \to \infty} \alpha_{i,n}(\hT) &\le \alpha\,.\label{eq:typeI}\\
\lim\inf_{n \to \infty} \frac{1-\beta_{i,n}(\hT;\lb)}{1-\beta_{i,n}^{*}(\lb)}& \ge 1\,,\;\;\;\;\;\;\; 
1-\beta_{i,n}^{*}(\lb) \equiv G\bigg(\alpha,\frac{\sqrt{n}\,\lb}{\sigma [\Sigma^{-1}_{i,i}]^{1/2}}\bigg)\,, \label{eq:power}
\end{align}
where, for $\alpha \in [0,1]$ and $u \in \reals_+$, the function $G(\alpha,u)$ is defined as follows:
\[
G(\alpha,u) = 2 - \Phi(\Phi^{-1}(1-\frac{\alpha}{2}) + u) - \Phi(\Phi^{-1}(1-\frac{\alpha}{2}) - u)\,. 
\]
\end{thm}

Theorem~\ref{thm:error-power} is proved in Appendix~\ref{proof:error-power}.
It is easy to see that, for any $\alpha >0$, $u \mapsto G(\alpha,u)$ is continuous and monotone increasing.
Moreover, $G(\alpha,0) = \alpha$ which is the trivial power obtained by randomly rejecting $H_{0,i}$ with probability
$\alpha$. As $\lb$ deviates from zero, we obtain nontrivial power. 
Notice that in order to achieve a specific power  $\beta>\alpha$, our scheme requires $\lb \ge c_\beta (\sigma /\sqrt{n})$, 
for some constant $c_\beta$ that depends on $\beta$.  This is because 
$\Sigma^{-1}_{i,i} \le \sigma_{\max}(\Sigma^{-1}) \le (\sigma_{\min}(\Sigma))^{-1}=  O(1)$.

\subsubsection{Near optimality of the hypothesis testing procedure}

The authors of~\cite{javanmard2013hypothesis} prove an upper bound for the minimax power of tests with
a given significance level $\alpha$, under random designs.
For the readers' convenience, we recall here this result.
(The following is a restatement of \cite[Theorem
2.3]{javanmard2013hypothesis}, together with a standard estimate on
the tail of chi-squared random variables.)
\begin{thm}[\cite{javanmard2013hypothesis}]\label{thm:GeneralUpperBound}
Assume $\bX\in\reals^{n\times p}$ to be a random design matrix with i.i.d. Gaussian rows
with zero mean and covariance $\Sigma$.
For $i\in [p]$, let $T_{i,\bX}:\reals^n\to\reals^n$  be a hypothesis
testing procedure for testing $H_{0,i}:\, \theta_{0,i}=0$, and denote
by $\alpha_i(T)$ and $\beta_{i,n}(T;\lb)$ its fraction of type I and
type II errors, cf. Eqs.~(\ref{eq:AlphaDef}) and (\ref{eq:BetaDef}).
Finally, for $S\subseteq [p]\setminus \{i\}$, define 
${\Sigma}_{i|S}
\equiv{\Sigma}_{ii}-{\Sigma}_{i,S}{\Sigma}_{S,S}^{-1}{\Sigma}_{S,i}\in\reals$.  

For any $\ell\in\reals$ and $|S| < s_0<n$, if 
$\alpha_{i,n}(T)\le \alpha$, then
\begin{align}
1 - \beta_{i,n}(T;\lb) & \le 
G\Big(\alpha,\frac{\lb}{\sigma_{\rm eff}(\xi)}\Big) 
+e^{-\xi^2/8}\,,\label{eq:UpperBoundPower}\\
\sigma_{\rm eff}(\xi) & \equiv
\frac{\sigma}{{\Sigma}^{1/2}_{i|S}(\sqrt{n-s_0+1}+\xi)}\, ,
\end{align}
for any  $\xi\in [0,(3/2)\sqrt{n-s_0+1}]$.
\end{thm}
The intuition behind this bound is straightforward: 
the power of any test for $H_{0,i}:\, \theta_{0,i}=0$ is upper bounded
by the power of an oracle test that is given access to the support of
$\theta_{0}$, with the eventual exclusion of $i$. Namely, the oracle
has
access to $\supp(\theta_0)\setminus\{i\}$ and outputs a test for
$H_{0,i}$. Computing the minimax power of such oracle reduces to a
classical hypothesis testing problem. 

Let us emphasize that the last theorem applies to \emph{Gaussian} random
designs. Since this theorem establishes a negative result (an upper
bound on power) it makes sense to consider this somewhat more
specialized setting.

Using this upper bound, we can restate Theorem \ref{thm:error-power}
as follows.
\begin{coro}\label{thm:optimality}
Consider a Gaussian random design model that satisfies the conditions of
Theorem~\ref{thm:error-power},
and let $\hT$ be the testing procedure defined in
Eq.~(\ref{eq:decision-rule}),
with $\htheta^u$ as in Algorithm 1.
Further, let
\begin{align}
\eta_{\Sigma,s_0}\equiv \min_{i \in [p];S} \Big\{\Sigma_{i|S}\, \Sigma_{ii}^{-1}: \,\, S \subseteq[p]\backslash\{i\},\, |S|< s_0 \Big\}\,.
\end{align}

Under the sparsity assumption $s_0 = o(\sqrt{n}/\log p)$, the
following holds true.
If $\{T_{i,\bX}\}$ is any sequence of tests with $\lim\sup_{n \to
  \infty} \alpha_{i,n}(T) \le \alpha$,
then 
\begin{align}
\lim\inf_{n \to \infty} \frac{1-\beta_{i,n}(\hT;\lb)}{1-\beta_{i,n/
  \eta_{\Sigma,s_0}}(T;\lb)}& \ge 1\,. 
\end{align}
In other words, the asymptotic efficiency of the test $\hT$ is at least $1/\eta_{\Sigma,s_0}$.
\end{coro}
Hence, our test $\hT$ has nearly  optimal power in the
following sense.
 It has power at least as large as the power of any
oter test $T$, provided the latter is applied to a sample size
increased  by  a factor $\eta_{\Sigma,s_0}$.

Further, under the assumptions of Theorem~\ref{thm:main_thm},
the factor $\eta_{\Sigma,s_0}$ is a bounded constant. Indeed
\begin{align}
\eta_{\Sigma,s_0} \le \Sigma^{-1}_{i,i} \Sigma_{i,i}
\le \frac{\sigma_{\max}(\Sigma)}{\sigma_{\min}(\Sigma)}\le\frac{C_{\rm
  max}}{C_{\rm min}}\,,
\end{align}
since $ \Sigma^{-1}_{ii} \le (\sigma_{\min}(\Sigma))^{-1}$, and 
$\Sigma_{i|S} \le \Sigma_{i,i} \le \sigma_{\max}(\Sigma)$ due to
$\Sigma_{S,S} \succ 0$. 

Note that $n$, $\gamma$ and $\sigma$ appears in our upper bound
(\ref{eq:UpperBoundPower})  in the
combination $\gamma\sqrt{n}/\sigma$, which is the natural measure of
the signal-to-noise ratio (where, for simplicity, we neglected $s_0=
o(\sqrt{n}/\log p)$ with respect to $n$). 
Hence, the above result can be restated as follows. The test $\hT$ has power at least as large as the power of any
oter test $T$, provided the latter is applied at a noise level
augmented by  a factor $\sqrt{\eta_{\Sigma,s_0}}$.

\subsection{Generalization to simultaneous confidence intervals}

In many situations, it is necessary to perform statistical inference 
on more than one of the parameters simultaneously. For instance,
we might be interested in performing inference about $\theta_{0,R}
\equiv (\theta_{0,i})_{i\in R}$ for some set $R\subseteq [p]$. 

The simplest generalization of our method is to the case in which
$|R|$ stays finite as $n,p\to\infty$. In this case we have the
following generalization of Lemma \ref{lemma:LastDistribution}.
(The proof is the same as for Lemma
\ref{lemma:LastDistribution},
and hence we omit it.)
\begin{lemma}\label{lemma:LowDim}
Under the assumptions of Lemma \ref{lemma:LastDistribution}, define
\begin{align}
Q^{(n)} \equiv \frac{\hsigma^2}{n} \, [M \hSigma M^\sT]\, .
\end{align}
Let $R = R(n)$ be a sequence of sets $R(n)\subseteq[p]$, with
$|R(n)|=k$ fixed as $n, p\to\infty$,  and further assume
$s_0=o(\sqrt{n}/\log p)$, with $s_0\ge 1$.
Then, for all $x=(x_1,\dots,x_k)\in\reals^k$, we have
\begin{eqnarray}\label{eq:distributionLowD}
\lim_{n\to\infty}\sup_{\theta_0\in\reals^p;\, \|\theta_0\|_0 \le s_0
}\left|\prob \left\{ (Q^{(n)}_{R,R})^{-1/2}(\htheta^u_R - \theta_{0,R})
\le x  \right\} -\Phi_k(x)\right| =0\, ,
\end{eqnarray}
where $(a_1,\dots,a_k)\le (b_1,\dots,b_k)$ indicates that $a_1\le
b_1$,\dots $a_k\le b_k$, and $\Phi_k(x) = \Phi(x_1)\cdots \Phi(x_k)$.
\end{lemma}
This lemma allows to construct confidence regions for low-dimensional
projections of $\theta_0$, much in the same way as we used Lemma
\ref{lemma:LastDistribution} to compute confidence intervals for
one-dimensional projections in Section \ref{sec:ConfidenceInterval}.

Explicitly, let $\cC_{k,\alpha}\subseteq\reals^k$ be any Borel set such
that $\int_{\cC_{k,\alpha}} \phi_k(x) \,\de x\ge 1-\alpha$ , where
$$\phi_k(x) = \frac{1}{(2\pi)^{k/2}}\,\exp\Big(-\frac{\|x\|^2}{2}\Big)\,,$$
is the $k$-dimensional Gaussian density. 
Then, for $R\subseteq [p]$, we define
$J_R(\alpha)\subseteq \reals^k$ as follows
\begin{align}
J_R(\alpha)\equiv \htheta_R^u+(Q^{(n)}_{R,R})^{1/2}\cC_{k,\alpha}\, .
\end{align}
Then Lemma \ref{lemma:LowDim}  implies (under the
assumptions stated there) that $J_R(\alpha)$ is a valid confidence region
\begin{align}
\lim_{n\to\infty}\prob\big(\theta_{0,R}\in J_{R}(\alpha)\big) =
1-\alpha\, .
\end{align}

A more challenging regime is the one of large-scale inference, that
corresponds to $|R(n)|\to\infty$ with $n$. Even in the seemingly simple case in
which a  correct $p$-value is given for each individual coordinate, the
problem of aggregating them has attracted considerable amount of work,
see e.g. \cite{efron2010large} for an overview. 

Here we limit ourself to designing a testing procedure for the family
of hypotheses $\{H_{0,i}: \, \theta_{0,i} = 0\}_{i\in [p]}$ that
controls the familywise error rate (FWER). Namely we want to define
$T_{i,\bX}:\reals^n\to \{0,1\}$, for each $i\in [p]$,
$\bX\in\reals^{n\times p}$ such that
\begin{align}
\FWER(T,n) \equiv \sup_{\theta_0\in\reals^p,\|\theta_0\|_0\le
  s_0}\prob\Big\{\exists i\in [p]:\;\; \theta_{0,i} = 0,
T_{i,\bX}(y)= 1\Big\}\, ,
\end{align}

In order to achieve familywise error control, we adopt a standard
trick based on Bonferroni inequality. Given $p$-values defined as per
Eq.~(\ref{eq:p-value}), we let
\begin{eqnarray}
\begin{split}\label{eq:decision-rule-FWER}
\hTf_{i,\bX}(y) = \begin{cases}
1 & \text{if } P_i \le \alpha/p \quad \quad \text{ (reject $H_{0,i}$)}\,,\\
0 & \text{otherwise} \quad\quad \text{(accept $H_{0,i}$)}\,.
\end{cases}
\end{split}
\end{eqnarray} 
Then we have the following error control guarantee.
\begin{thm}\label{thm:FWER}
Consider a sequence of design matrices $\bX\in\reals^{n\times p}$,
 with dimensions $n\to\infty$, $p=p(n)\to\infty$ satisfying the
 assumptions of Lemma \ref{lemma:LastDistribution}.

Consider the linear model~\eqref{eqn:regression} and let $\htheta^u$ be defined as per
Eq.~\eqref{eq:hthetau} in Algorithm 1, with $\coh =a\sqrt{(\log p)/n}$
and $\lambda = \sigma\sqrt{(c^2\log p)/n}$, with $a,c$  large enough constants.
Finally, let $\hsigma = \hsigma(y,\bX)$ be a consistent estimator of the noise
level in the sense of Eq.~(\ref{eq:ConsistencySigma}),
and $\hT$ be the test defined in Eq.~(\ref{eq:decision-rule-FWER}).
Then:
\begin{align}
\underset{n\to\infty}{\lim\sup}\,\, \FWER(\hTf,n) \le \alpha\, .
\end{align}
\end{thm}
The proof of this theorem is similar to the one of Lemma
\ref{lemma:LastDistribution} and Theorem \ref{thm:error-power}, and is
deferred to Appendix \ref{app:FWER}.
%
%
\section{Non-Gaussian noise}
\label{sec:NonGaussian}

As can be seen from the proof of Theorem~\ref{thm:main_thm}, $Z = M \bX^\sT W/\sqrt{n}$,
and since the noise is Gaussian, i.e., $W \sim \normal(0,\sigma^2
\id)$, we have $Z|\bX \sim \normal(0,\sigma^2 M\hSigma M^\sT)$.
We claim that the distribution of the coordinates of $Z$ is
asymptotically Gaussian, even if $W$ is non-Gaussian, provided the
definition of $M$ is modified slightly. As a consequence, the
definition of confidence intervals and $p$-values in Corollary
\ref{coro:Interval} and~\eqref{eq:p-value} remain valid in this
broader setting.

In case of non-Gaussian noise, we write 
\begin{align*}
\frac{\sqrt{n}(\htheta^u_i - \theta_{0,i})}{\sigma [M\hSigma M^\sT]^{1/2}_{i,i}} 
& = \frac{1}{\sqrt{n}} \frac{m_i^\sT \bX^\sT W}{\sigma [m_i^\sT \hSigma m_i]^{1/2}} + o(1)\\
& = \frac{1}{\sqrt{n}} \sum_{j=1}^n \frac{m_i^\sT X_j W_j}{\sigma [m_i^\sT \hSigma m_i]^{1/2}} +o(1)\,.
\end{align*}
Conditional on $\bX$, the summands $\xi_j = m_i^\sT X_j W_j /(\sigma
[m_i^\sT \hSigma m_i]^{1/2})$ 
are independent and zero mean.
Further, $\sum_{j=1}^n \E(\xi_j^2| \bX) = 1$.
Therefore, if Lindenberg condition holds, namely for every $\eps > 0$,
almost surely
\begin{eqnarray}
\lim_{n\to \infty} \frac{1}{n} \sum_{j=1}^n \E(\xi_j^2 \ind_{\{|\xi_j| > \eps \sqrt{n}\}}|\bX) =0\,,\label{eq:lindenberg}
\end{eqnarray}
then $\sum_{j=1}^n \xi_j /\sqrt{n}|\bX \dist \normal(0,1)$, from which we can build the valid $p$-values as in~\eqref{eq:p-value}.

In order to ensure that the Lindeberg condition holds, we modify the optimization problem~\eqref{eq:optimization_mod} as follows:
\begin{eqnarray}
\begin{split}\label{eq:optimization_mod}
&\text{minimize } \quad \, m^\sT \hSigma m\\
&\text{subject to} \quad \|\hSigma m - e_i \|_{\infty} \le \coh\\
& \hspace{2.1cm} \|\bX m\|_\infty \le n^{\beta} \quad \text{for arbitrary fixed }0< \beta< 1/2
\end{split}
\end{eqnarray}
Next theorem shows the validity of the proposed $p$-values in the non-Gaussian noise setting.
\begin{thm}\label{thm:nongauss}
Suppose that the noise variables $W_i$ are independent with $\E(W_i) =
0$,  $\E(W_i^2) = \sigma^2$, and $\E(|W_i|^{2+a})\le C\, \sigma^{2+a}$
for some $a>(1/2-\beta)^{-1}$.

Let $M = (m_1,\dotsc,m_p)^\sT$ be the matrix with rows $m_i^\sT$ obtained by solving optimization 
problem~\eqref{eq:optimization_mod}. Then under the assumptions of
Theorem~\ref{thm:main_thm}, and for
sparsity level $s_0 = o(\sqrt{n}/\log p)$,  an asymptotic two-sided
confidence interval for $\theta_{0,i}$ with significance $\alpha$
is given by $I_i = [\htheta^u_{i} -\delta(\alpha,n), \htheta^u_{i}
+\delta(\alpha,n)]$ where
\begin{align}
 \delta(\alpha,n) = \Phi^{-1}(1-\alpha/2)  \hsigma\, n^{-1/2} \sqrt{[M \hSigma M^\sT]_{i,i}}\,.
 \end{align}
Further, an asymptotically valid $p$-value $P_i$ for testing null hypothesis $H_{0,i}$ is constructed as:
\[
P_i = 2\bigg(1- \Phi\bigg(\frac{\sqrt{n}|\htheta^u_i|}{[M\hSigma M^\sT]_{i,i}^{1/2}}\bigg) \bigg)\,.
\]
\end{thm} 
Theorem~\ref{thm:nongauss} is proved in Section~\ref{proof:nongauss}.

\section{Numerical experiments}\label{sec:simulation}
\subsection{Synthetic data}\label{sec:synthetic}
We consider linear model~\eqref{eq:NoisyModel}, where the rows of design matrix $\bX$ are fixed i.i.d. realizations
from $\normal(0,\Sigma)$, where $\Sigma \in \reals^{p\times p}$ is a circulant symmetric matrix with 
entries $\Sigma_{jk}$ given as follows for $j\le k$:
\begin{eqnarray}
\Sigma_{jk} = \begin{cases}
1& \text{if } k=j\,,\\
0.1 & \text{if } k\in \{j+1,\dotsc,j+5\}\\
& \text{or } k\in \{j+p-5, \dotsc, j+p-1\}\,,\\
0 & \text{for all other } j\le k\,.
\end{cases}
\end{eqnarray}
Regarding the regression coefficient, we consider a uniformly random support $S \subseteq [p]$,
with $|S| = s_0$ and let $\theta_{0,i} = b$ for $i \in S$ and $\theta_{0,i} = 0$ otherwise. 
The measurement errors are $W_i\sim \normal(0,1)$, for $i\in [n]$.
We consider several configurations of $(n,p,s_0, b)$ and for each configuration report
our results based on $20$ independent realizations of the model with fixed design and
fixed regression coefficients. In other words, we repeat experiments over $20$ independent
realization of the measurement errors. 

We use the regularization parameter $\lambda = 4\hsigma\sqrt{(2\log p)/n}$, where $\hsigma$
is given by the scaled LASSO as per equation~\eqref{eq:SLASSO} with $\widetilde{\lambda} =  10\sqrt{(2\log p)/n}$.
Furthermore, parameter $\mu$ (cf. Eq.~\eqref{eq:optimization}) is set to
\begin{eqnarray*}
\mu = 2 \sqrt{\frac{\log p}{n}}\,.
\end{eqnarray*}
This choice of $\mu$ is guided by Theorem~\ref{thm:event_thm} $(b)$.
%
%

Throughout, we set the significance level $\alpha = 0.05$. 

\smallskip
\noindent{{\bf Confidence intervals.}}
For each configuration, we consider $20$ independent realizations of measurement noise and for each
parameter $\theta_{0,i}$, we compute the average length of the corresponding confidence interval, denoted by
$\avglength (J_i(\alpha))$ where $J_i(\alpha)$ is given by equation~\eqref{eq:CI} and the 
average is taken over the realizations. We then define
\begin{eqnarray}
\ell \equiv p^{-1} \sum_{i\in [p]} \avglength(J_i(\alpha))\,.
\end{eqnarray}
We also consider the average length of intervals for the active and inactive parameters, as follows:
\begin{align}
\ell_S \equiv s_0^{-1} \sum_{i\in S} \avglength(J_i(\alpha))\,, \quad 
\ell_{S^c} \equiv (p-s_0)^{-1} \sum_{i\in S^c} \avglength(J_i(\alpha))\,.
\end{align}

Similarly, we consider average coverage for individual parameters.
We define the following three metrics:
\begin{align}
\cov &\equiv p^{-1} \sum_{i\in [p]} \hprob[\theta_{0,i} \in J_i(\alpha)]\,,\\
\cov_S &\equiv s_0^{-1} \sum_{i\in S} \hprob[\theta_{0,i} \in J_i(\alpha)]\,,\\
\cov_{S^c} &\equiv (p-s_0)^{-1} \sum_{i\in S^c} \hprob[0 \in J_i(\alpha)]\,,
\end{align}
where $\hprob$ denotes the empirical probability computed based on the $20$ realizations for each configuration.
The results are reported in Table~\ref{tbl:confidence}.
In Fig.~\ref{fig:CI}, we plot the constructed $95\%$-confidence intervals for one realization of configuration $(n,p,s_0,b) = (1000,600,10,1)$.
For sake of clarity, we plot the confidence intervals for only 100 of the 1000 parameters.

\begin{figure}[]
\centering
\includegraphics*[width = 3.5in]{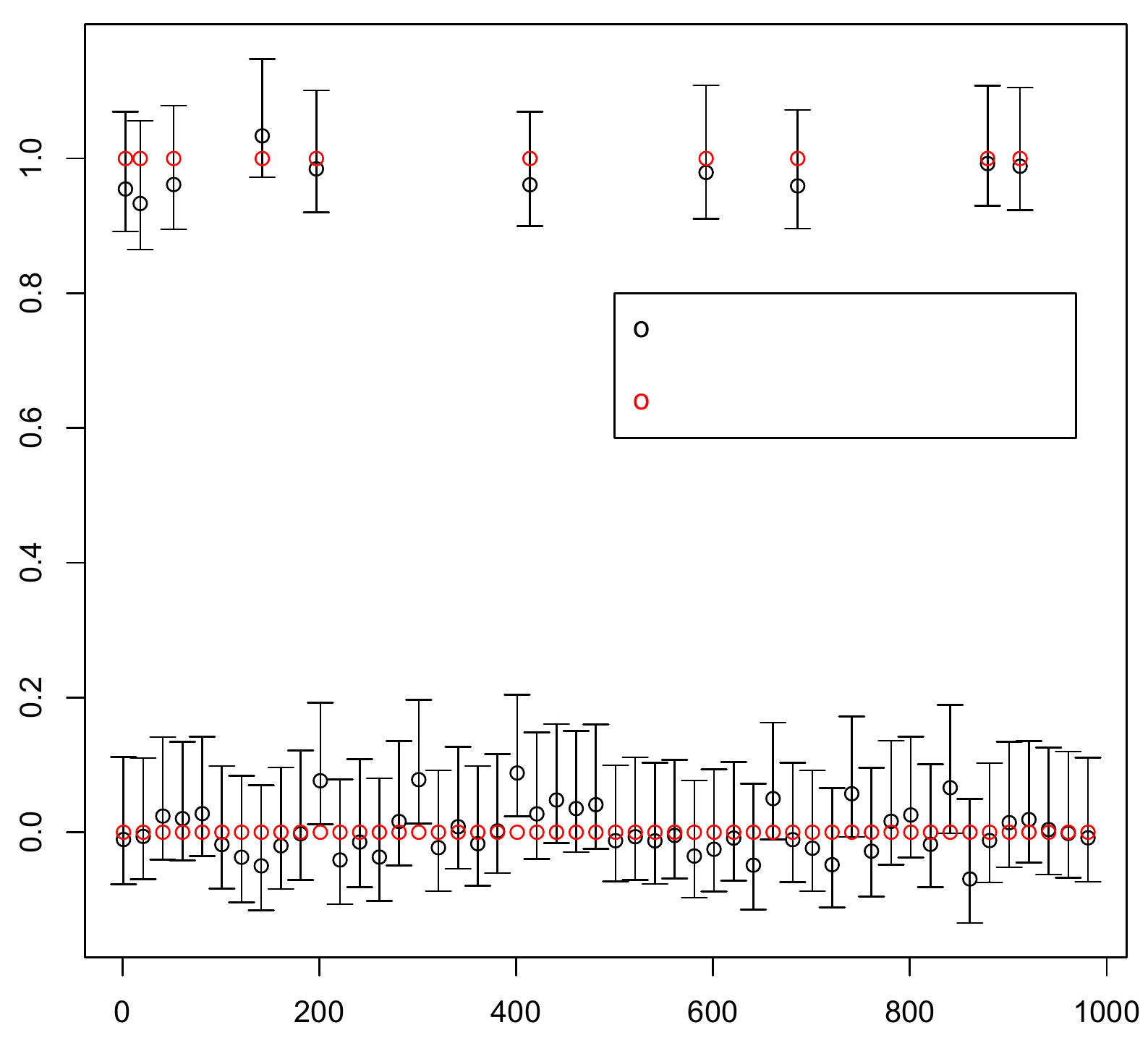}
\put(-105,155){{\small coordinates of $\theta_{0}$}}
\put(-105,138){{\small coordinates of $\htheta^u$}}
\caption{ $95\%$ confidence intervals for one realization of configuration $(n,p,s_0,b) = (1000,600,10,1)$.
For clarity, we plot the confidence intervals for only 100 of the 1000 parameters. The true parameters
$\theta_{0,i}$ are in red and the coordinates of the debiased estimator $\htheta^u$ are in black.}\label{fig:CI}
\end{figure}

\begin{table*}[]
\begin{center}
{\small
\begin{tabular}{|l|c|c|c|c|c|c|c|c|}\hline
\diaghead{\theadfont Configuration Measure}%
{Configuration}{Measure}&
\thead{$\ell$}&\thead{$\ell_S$} & \thead{$\ell_{S^c}$} & \thead{$\cov$} & \thead{$\cov_S$} & \thead{$\cov_{S^c}$}\\
\hline
{$(1000, 600, 10, 0.5)$} &  0.1870 & 0.1834 & 0.1870 & 0.9766 &0.9600 &0.9767 \\
{$(1000, 600, 10, 0.25)$} &  0.1757 & 0.1780 & 0.1757 & 0.9810 &0.9000 &0.9818 \\
{$(1000, 600, 10, 0.1)$} &  0.1809 & 0.1823 & 0.1809 & 0.9760 & 1 &0.9757 \\
{$(1000, 600, 30, 0.5)$} &  0.2107 & 0.2108 & 0.2107 & 0.9780 &0.9866 &0.9777 \\
{$(1000, 600, 30, 0.25)$} &   0.1956 & 0.1961 & 0.1956 & 0.9660 &0.9660 &0.9659 \\
{$(1000, 600, 30, 0.1)$} &   0.2023 & 0.2043 & 0.2023 & 0.9720 &0.9333 &0.9732 \\
{$(2000, 1500, 50, 0.5)$} &  0.1383 & 0.1391 & 0.1383 & 0.9754 & 0.9800 &0.9752 \\
{$(2000, 1500, 50, 0.25)$} &  0.1356 & 0.1363 & 0.1355 & 0.9720 & 0.9600 & 0.9723 \\
{$(2000, 1500, 50, 0.1)$} & 0.1361  &  0.1361&  0.1361&  0.9805& 1 & 0.9800 \\
{$(2000, 1500, 25, 0.5)$} &  0.1233 & 0.1233 & 0.1233 & 0.9731 & 0.9680 & 0.9731 \\
{$(2000, 1500, 25, 0.25)$} &  0.1208 & 0.1208 & 0.1208 & 0.9735 & 1 & 0.9731 \\
{$(2000, 1500, 25, 0.1)$} & 0.1242  & 0.1237 & 0.1242 &  0.9670& 0.9200 & 0.9676 \\
\hline
\end{tabular}
}
\end{center}
\caption{Simulation results for the synthetic data described in Section~\ref{sec:synthetic}. The results corresponds
to $95\%$ confidence intervals. }\label{tbl:confidence}
\end{table*}

\smallskip
\noindent{\bf False positive rates and statistical powers.}
Table~\ref{tbl:FPTP} summarizes the false positive rates and the statistical powers achieved
by our proposed method, the multisample-splitting method~\cite{multispliting}, and the ridge-type projection estimator~\cite{BuhlmannSignificance} for several configurations. The results are obtained by taking average
over $20$ independent realizations of measurement errors for each configuration. 
As we see the multisample-splitting achieves false positive rate 0 on all of the configurations considered here, making no type I error. However, the true positive rate is always smaller than that of our proposed method.
By contrast, our method achieves false positive rate close to the pre-assigned significance level $\alpha = 0.05$ and 
obtains much higher true positive rate. Similar to the multisample-splitting, the ridge-type projection estimator is conservative and achieves false positive rate smaller
than $\alpha$. This, however, comes at the cost of a smaller true positive rate than our method.
It is worth noting that an ideal testing procedure should allow to control the level of statistical significance $\alpha$, and obtain the 
maximum true positive rate at that level. 

Here, we used the {\sf R}-package {\sf hdi} to test multisample-splitting and the ridge-type projection estimator.  

Let $Z= (z_i)_{i=1}^p$ denote the vector with $z_i \equiv \sqrt{n}(\htheta^u_i - \theta_{0,i})/\hsigma \sqrt{[M\hSigma M^\sT]_{i,i}}$.
Fig.~\ref{fig:Z_qqnorm} shows the sample quantiles of $Z$ versus the quantiles of the standard normal distribution
for one realization of the configuration $(n,p,s_0,b) = (1000,600,10,1)$. 
The scattered points are close to the line with unit slope and zero intercept. This confirms the result of Theorem~\ref{lemma:LastDistribution}
regarding the gaussianity of the entries $z_i$.

For the same problem, in Fig.~\ref{fig:p-value} we plot the empirical CDF of the computed $p$-values restricted to the variables outside the support. Clearly, the $p$-values for these entries are uniformly distributed as expected.

\begin{table*}[h]
\begin{center}
{\small
\begin{tabular}{|c|c|c|c|c|c|c| }
\hline
\multicolumn{1}{ |c| }{} & \multicolumn{2}{ c| }{Our method} & \multicolumn{2}{ c| }{Multisample-splitting}
& \multicolumn{2}{ c| }{Ridge-type projection estimator}\\
\hline
Configuration& FP & TP & FP & TP & FP & TP\\ \hline
{$(1000, 600, 10, 0.5)$} &0.0452 & 1 & 0& 1&0.0284 & 0.8531\\
{$(1000, 600, 10, 0.25)$}&0.0393 & 1& 0& 0.4& 0.02691& 0.7506 \\
{$(1000, 600, 10, 0.1)$} &0.0383 & 0.8&0 &0 & 0.2638 & 0.6523 \\
{$(1000, 600, 30, 0.5)$} &0.0433 & 1& 0&1 &0.0263 & 0.8700\\
{$(1000, 600, 30, 0.25)$}&0.0525 &1 & 0&0.4 &0.2844 & 0.8403  \\
{$(1000, 600, 30, 0.1)$} & 0.0402& 0.7330&0 &0 &0.2238 &0.6180  \\
{$(2000, 1500, 50, 0.5)$} &0.0421 & 1&0 & 1&0.0301 &0.9013 \\
{$(2000, 1500, 50, 0.25)$} & 0.0415& 1& 0& 1& 0.0292&0.8835  \\
{$(2000, 1500, 50, 0.1)$} & 0.0384& 0.9400&0 &0 & 0.02655& 0.7603  \\
{$(2000, 1500, 25, 0.5)$} &0.0509 & 1&0 & 1& 0.0361& 0.9101 \\
{$(2000, 1500, 25, 0.25)$} &0.0481 &1 &0 &1 & 0.3470&0.8904  \\
{$(2000, 1500, 25, 0.1)$} & 0.0551& 1& 0& 0.16&0.0401 & 0.8203 \\
\hline
\end{tabular}
}
\end{center}
\caption{Simulation results for the synthetic data described in Section~\ref{sec:synthetic}. The false positive rates (FP) and 
the true positive rates (TP) are computed at significance level $\alpha = 0.05$.}\label{tbl:FPTP}
\end{table*}

\begin{figure}[]
\centering
\includegraphics*[width = 2.8in]{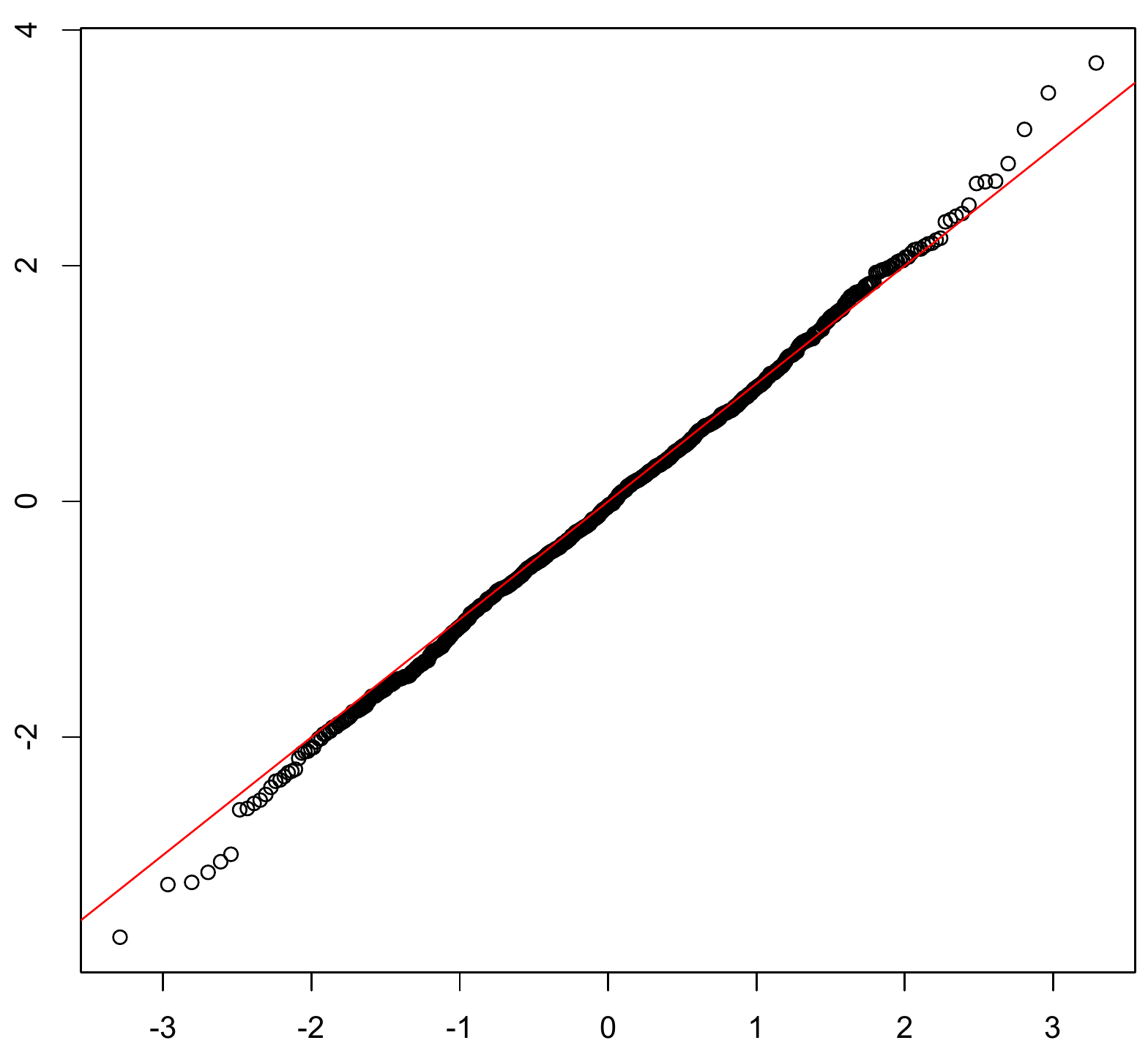}
\put(-187,-15){{\small Quantiles of  standard normal distribution}}
\put(-220,50){\rotatebox{90}{\small Sample quantiles of $Z$}}
\caption{Q-Q plot of $Z$ for one realization of configuration $(n,p,s_0,b) = (1000,600,10,1)$.}\label{fig:Z_qqnorm}
\end{figure}

\begin{figure}[]
\centering
\includegraphics*[width = 2.8in]{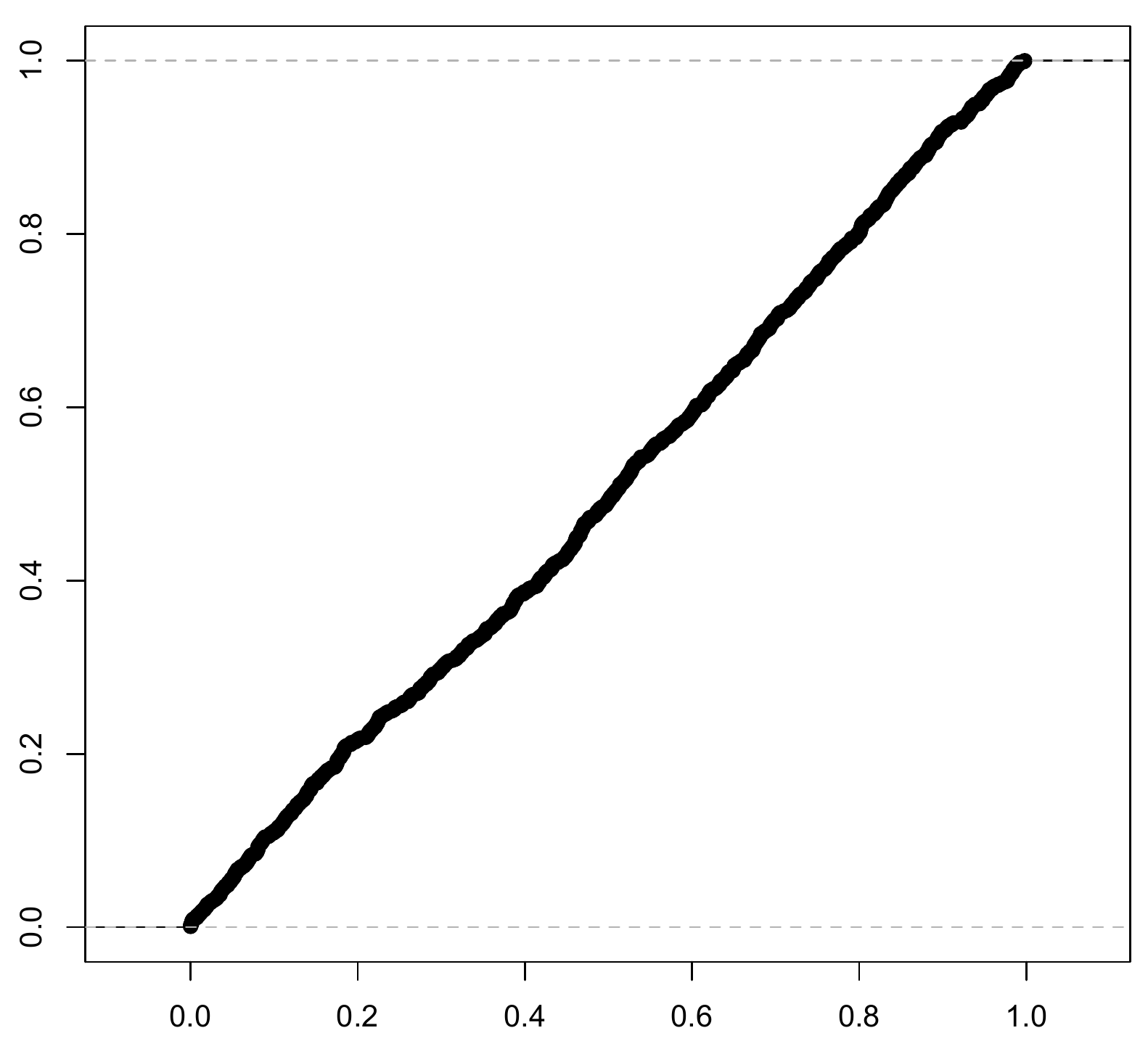}
\caption{Empirical CDF of the computed $p$-values (restricted to entries outside the support)
for one realization of configuration $(n,p,s_0,b) = (1000,600,10,1)$. Clearly the plot confirms that the $p$-values
 are distributed according to uniform distribution.}\label{fig:p-value}
\end{figure}   
\subsection{Real data}
As a real data example, we consider a high-throughput genomic data set concerning riboflavin (vitamin $B_2$) production rate.
This data set is made publicly available by~\cite{BuhlmannBio} and contains $n=71$ samples and $p = 4,088$ covariates corresponding to $p = 4,088$ genes.
For each sample, there is a real-valued response variable indicating the logarithm of the riboflavin production rate along with
the logarithm of the expression level of the $p = 4,088$ genes as the covariates. 

Following~\cite{BuhlmannBio}, we model the riboflavin production rate as a linear model with $p = 4,088$ covariates
and $n = 71$ samples, as in Eq.~\eqref{eqn:regression}. We use the ${\sf R}$ package ${\sf glmnet}$~\cite{glmnet} to fit the LASSO estimator.
Similar to the previous section, we use the regularization parameter $\lambda = 4\hsigma\sqrt{(2\log p)/n}$, where $\hsigma$
is given by the scaled LASSO as per equation~\eqref{eq:SLASSO} with $\widetilde{\lambda} =  10\sqrt{(2\log p)/n}$.
This leads to the choice $\lambda = 0.036$. The resulting model contains 30 genes (plus an intercept term) corresponding to the nonzero parameters of the lasso 
estimator.

We use Eq.~\eqref{eq:p-value} to construct $p$-values for different genes. Adjusting FWER to $5\%$ significance level, we find two 
significant genes, namely genes {\sf YXLD-at} and {\sf YXLE-at}. 
By contrast, the multisample-splitting method proposed in~\cite{multispliting} finds only the gene {\sf YXLD-at} at the FWER-adjusted $5\%$ significance level. Also the Ridge-type projection estimator, proposed in~\cite{BuhlmannSignificance}, returns no significance gene. (See~\cite{BuhlmannBio} for further discussion on these methods.) This indicates that these methods are more conservative and produce typically larger $p$-values.  

In Fig.~\ref{fig:p-Rib} we plot the empirical CDF of the computed $p$-values 
for riboflavin example. Clearly the plot confirms that the $p$-values
 are distributed according to uniform distribution.
 
\begin{figure}[h]
\centering
\includegraphics*[width = 2.8in]{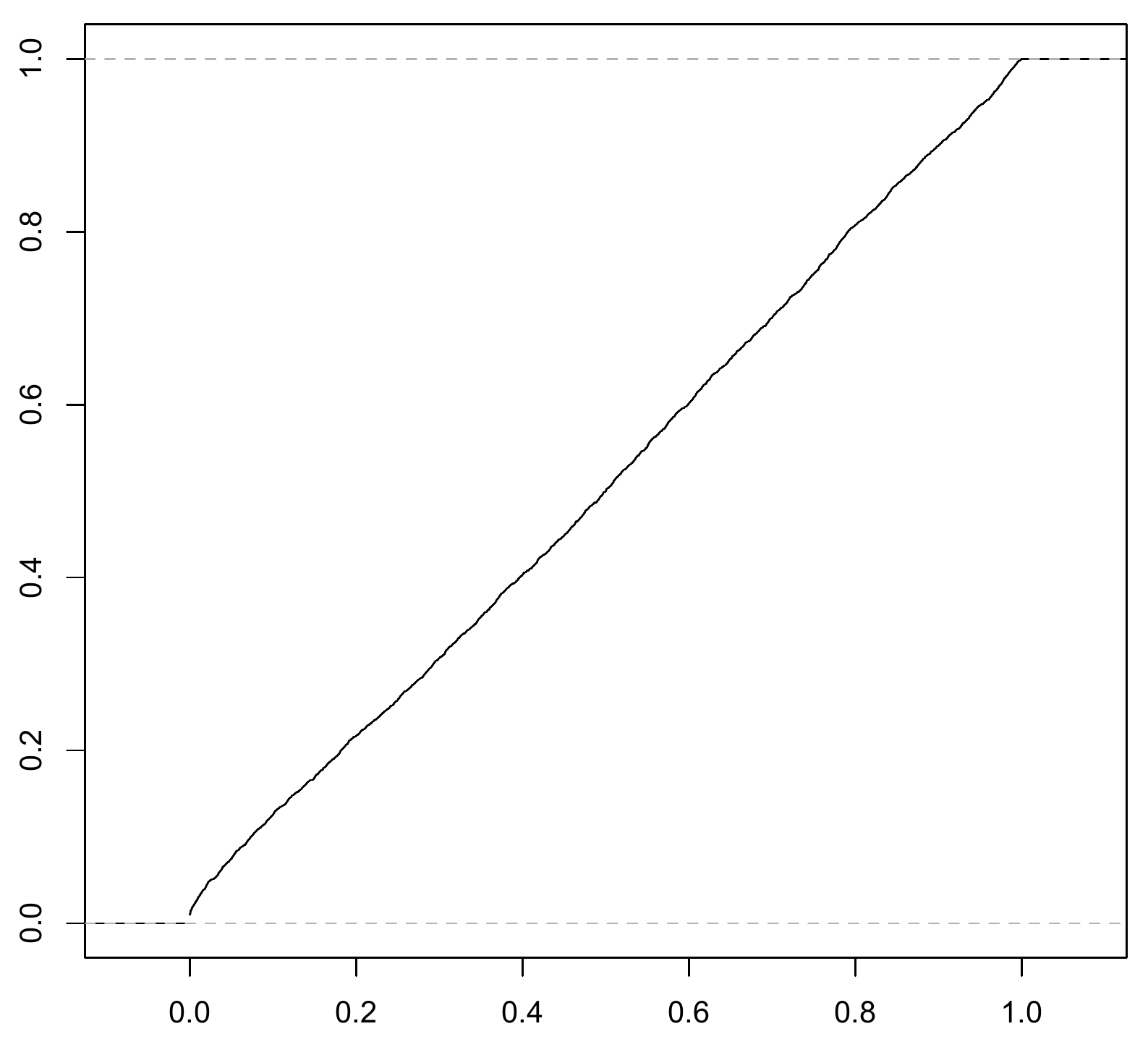}
\put(-187,-15){{\small Quantiles of  standard normal distribution}}
\put(-220,50){\rotatebox{90}{\small Sample quantiles of $Z$}}
\caption{Empirical CDF of the computed $p$-values 
for riboflavin example. Clearly the plot confirms that the $p$-values
 are distributed according to uniform distribution.}\label{fig:p-Rib}
\end{figure}   
\section{Proofs}

\subsection{Proof of Theorem~\ref{thm:deterministic}}
\label{sec:ProofDeterministic}

Substituting $Y=\bX\theta_0+W$ in the definition
(\ref{eq:GeneralDebiased}),
we get 
\begin{align}
\htheta^* &=
\htheta^n+\frac{1}{n}M\bX^{\sT}\bX(\theta_0-\htheta^n)+\frac{1}{n}M\bX^{\sT}
W\\
&= \theta_0+\frac{1}{\sqrt{n}}\, Z +\frac{1}{\sqrt{n}}\, \Delta\, ,
\end{align}
with $Z,\Delta$ defined as per the theorem statement. Further $Z$ is
Gaussian with the stated covariance because it is a linear function of
the Gaussian vector $W\sim \normal(0,\sigma^2\,\id_{p\times p})$. 

We are left with the task of proving the bound
(\ref{eq:FirstBoundDelta}) on $\Delta$. Note that by
definition (\ref{def:Coherence}), we have
\begin{align}
\|\Delta\|_{\infty}\le \sqrt{n}\,
|M\hSigma-\id|_{\infty}\,\|\htheta^n-\theta_0\|_1 = \sqrt{n}\, \mu_*
\|\htheta^n-\theta_0\|_1 \, . \label{eq:LinftyL1}
\end{align}
By \cite[Theorem 6.1, Lemma 6.2]{buhlmann2011statistics}, we have, for any $\lambda\ge 4\sigma\sqrt{2K\log(pe^{t^2/2})/n}$
\begin{align}
\prob\Big(\|\htheta^n-\theta_0\|_1\ge
\frac{4\lambda s_0}{\phi_0^2}\Big)\le 2\, e^{-t^2/2}\, .
\end{align}
(More precisely, we consider the trivial generalization of
\cite[Lemma 6.2]{buhlmann2011statistics} to the case
$(\bX^T\bX/n)_{ii}\le K$, instead of $(\bX^T\bX/n)_{ii}=1$ for all
$i\in [p]$.)

Substituting Eq.~(\ref{eq:LinftyL1}) in the last bound, we get 
\begin{align}
\prob\Big(\|\Delta\|_\infty\ge
\frac{4\lambda \mu_* s_0\sqrt{n}}{\phi_0^2}\Big)\le 2\, e^{-t^2/2}\, .
\end{align}
Finally, the claim follows by selecting $t$ so that $e^{t^2/2} = p^{c_0}$.

\subsection{Proof of Theorem~\ref{thm:event_thm}.$(a)$}
\label{proof:thm_eventA}

Note that the event $\cE_n$ requires two conditions. Hence, its complement
\begin{align}
\event_n(\phi_0,s_0,K)^c &= \cB_{1,n}(\phi_0,s_0)\cup\cB_{2,n}(K)\, ,\\
\cB_{1,n}(\phi_0,s_0) & \equiv\Big\{\bX\in\reals^{n\times
  p}:\,\;\min_{S:\; |S|\le s_0}\phi(\hSigma,S)<\phi_0 , \;\; \hSigma = (\bX^{\sT}\bX/n)\Big\}\, ,\\
\cB_{2,n}(K) & \equiv\Big\{\bX\in\reals^{n\times
  p}:\max_{i\in [p]}\, \hSigma_{i,i} \le K, \;\; \hSigma = (\bX^{\sT}\bX/n)\Big\}\, .
\end{align}
We will bound separately the probability of  $\cB_{1,n}$ and the
probability of $\cB_{2,n}$. The claim of
Theorem~\ref{thm:event_thm}.$(a)$ follows by union bound.
\subsubsection{Controlling $\cB_{1,n}(\phi_0,s_0)$}

It is also useful to recall the notion of
restricted eigenvalue, introduced by Bickel, Ritov and
Tsybakov \cite{BickelEtAl}.
\begin{definition}
Given a symmetric matrix $Q\in\reals^{p\times p}$ an integer
$s_0\ge 1$, and $L>0$, the restricted eigenvalue of $Q$ is
defined as
\begin{align}
\re^2(Q,s_0,L) \equiv\min_{S\subseteq [p], |S|\le s_0}\min_{\theta\in\reals^p}\Big\{\frac{\<\theta,Q\,\theta\>}{\|\theta_S\|_2^2} :\;\;
\theta\in\reals^p, 
\;\; \|\theta_{S^c}\|_1\le L\|\theta_S\|_1\Big\}\, .
\end{align} 
\end{definition}

Rudelson and Zhou \cite{rudelson2011reconstruction} prove that, if
the population covariance satisfies the restricted eigenvalue
condition, then the sample covariance satisfies it as well, with high
probability.
More precisely \cite[Theorem 6]{rudelson2011reconstruction}, the following happens for some $c_*\le 2000$, $m\equiv c_*
s_0C_{\rm max}^2 /\re^2(\Sigma,s_0,9)$, and
every $n\ge 4c_*m\kappa^4\log(60ep/(m\kappa))$ we have
\begin{align}
\prob\Big(\re(\hSigma,s_0,3)\ge \frac{1}{2}\re(\Sigma,s_0,9)\Big)\ge 1-2e^{-n/(4c_*\kappa^4)}
\end{align}
Note  that $\re(\Sigma,s_0,9)\ge \sigma_{\rm
  min}(\Sigma)^{1/2}\ge C_{\rm min}^{1/2}$ and, by Cauchy-Schwartz
$\min_{S:|S|\le s_0}\phi(\hSigma,S)\ge \re(\hSigma,s_0,3)$. 
With the definitions in the statement (cf. Eq.~(\ref{eq:BoundEvA})), we therefore have
\begin{align}
\prob\Big(\min_{S:|S|\le s_0}\phi(\hSigma,S)\ge \frac{1}{2}C_{\rm
  min}^{1/2}\Big)\ge 1-2e^{-c_1n}\, .
\end{align}
Equivalently, $\prob(\cB_{1,n}(\phi_0,s_0))\le 2\, e^{-c_1n}$.

\subsubsection{Controlling $\cB_{2,n}(K)$}

By definition
\begin{align}
\hSigma_{ii}-1 = \frac{1}{n}\sum_{\ell=1}^n(\<X_\ell,e_i\>^2-1)
=\frac{1}{n}\sum_{\ell=1}^n u_{\ell} ,\, .
\end{align}
Note that $u_{\ell}$ are independent centered random
variables. Further (recalling that, for any random variables $U, V$,
$\|U+V\|_{\psi_1}\le \|U\|_{\psi_1}+\|V\|_{\psi_1}$, and
$\|U^2\|_{\psi_1}\le 2\|U\|_{\psi_2}^2$) they are subexponential with subexponential norm
\begin{align}
\|u_{\ell}\|_{\psi_1}&\le 2\|\<X_\ell,e_i\>^2\|_{\psi_1}\le
4\|\<X_\ell,e_i\>\|^2_{\psi_1}\\
&\le 4\|\<\Sigma^{-1/2}X_{\ell},\Sigma^{1/2}e_i\>\|_{\psi_1}^2 \\
&\le
4\kappa^2 \|\Sigma^{1/2}e_i\|_2^2 = 4\kappa^2\Sigma_{ii} = 4\kappa^2\, .
\end{align}
 By Bernstein-type inequality for centered subexponential
random variables \cite{Vershynin-CS}, we get 
\begin{align}
\prob\Big\{\frac{1}{n} \Big|\sum_{\ell=1}^n u_{\ell} \Big| \ge \eps
\Big\} \le 
2 \exp \Big[ -\frac{n}{6} \min\Big((\frac{\eps}{4e\kappa^2})^2, \frac{\eps}{4e\kappa^2}\Big) \Big]\,.
\end{align}
Hence, for all $\eps$ such that
$\eps/(e\kappa^2)\in[\sqrt{(48\log p)/n}, 4]$,
\begin{align}
\prob\Big(\max_{i\in[p]}\hSigma_{ii}\ge 1+\eps\Big)&\le 2 p\, 
\exp\Big(-\frac{n\eps^2}{24e^2\kappa^4} \Big)
\le 2e^{-c_1n}\, ,
\end{align}
which implies $\prob(\bX\in\cB_{2,n}(K))\le 2\, e^{-c_1n}$ for all
$K-1\ge 20\kappa^2\sqrt{(\log p)/n}\ge \sqrt{(48e^2\kappa^4\log p)/n}$.

\subsection{Proof of Theorem~\ref{thm:event_thm}.$(b)$}
\label{proof:thm_eventB}

Obviously, we have
\begin{align}
\com(\bX)\le \big|\Sigma^{-1}\hSigma-\id\big|\, ,
\end{align}
and hence the statement follows immediately from the following estimate.
\begin{lemma}\label{lem:main_lem}
Consider a random design matrix $\bX \in \reals^{p\times p}$, with i.i.d. rows  having mean zero and population covariance $\Sigma$.
Assume that
\begin{itemize}
\item[$(i)$] We have $\sigma_{\min} (\Sigma) \ge C_{\min} >0$, and $\sigma_{\max}(\Sigma) \le C_{\max} <\infty$.
\item[$(ii)$] The rows of $X\Sigma^{-1/2}$ are sub-gaussian with $\kappa = \|\Sigma^{-1/2} X_1\|_{\psi_2}$.
\end{itemize}
Let $\hSigma = (\bX^\sT \bX)/n$ be the empirical covariance. Then, for any constant
$C > 0$, the following holds true.
\begin{eqnarray}
\prob \bigg\{\Big|\Sigma^{-1} \hSigma - \id \Big|_\infty \ge a \sqrt{\frac{\log p}{n}} \bigg\} \le 2p^{-c_2}\,, 
\end{eqnarray}
with $c_2 = (a^2 C_{\min})/(24 e^2 \kappa^4 C_{\max}) - 2$.
\end{lemma}
\begin{proof}[Proof of Lemma \ref{lem:main_lem}]
The proof is based on Bernstein-type inequality for sub-exponential random variables~\cite{Vershynin-CS}.
Let $\tilde{X}_\ell = \Sigma^{-1/2} X_\ell$, for $\ell \in [n]$, and write
\[
Z\equiv \Sigma^{-1} \hSigma - \id = \frac{1}{n} \sum_{\ell=1}^n \Big\{\Sigma^{-1} X_\ell X_\ell^\sT - \id \Big\}
= \frac{1}{n} \sum_{\ell=1}^n \Big\{\Sigma^{-1/2} \tilde{X}_\ell \tilde{X}_\ell^\sT \Sigma^{1/2} - \id \Big\}\,.
\] 
Fix $i,j \in [p]$, and for $\ell \in [n]$, let 
$v^{(ij)}_\ell = \<\Sigma^{-1/2}_{i,\cdot}, \tilde{X}_\ell \> \<\Sigma^{1/2}_{j,\cdot}, \tilde{X}_\ell \> - \delta_{i,j}$, where $\delta_{i,j} = {\bf 1}_{\{i=j\}}$.
Notice that $\E(v^{(ij)}_\ell) = 0$, and the $v^{(ij)}_\ell$ are independent for $\ell \in [n]$. Also, $Z_{i,j} = (1/n) \sum_{\ell=1}^n v^{(ij)}_\ell$.
By~\cite[Remark 5.18]{Vershynin-CS}, we have 
\[
\|v^{(ij)}_\ell\|_{\psi_1} \le 2 \|\<\Sigma^{-1/2}_{i,\cdot}, \tilde{X}_\ell \> \<\Sigma^{1/2}_{j,\cdot}, \tilde{X}_\ell \> \|_{\psi_1}.
\]
Moreover, for any two random variables $X$ and $Y$, we have
\begin{align*}
\|XY\|_{\psi_1} &= \sup _{p \ge 1} \,  p^{-1} \E(|XY|^p)^{1/p} \\
& \le \sup _{p \ge 1} \, p^{-1} \E(|X|^{2p})^{1/2p}\, \E(|Y|^{2p})^{1/2p} \\
&\le 2\, \Big(\sup _{q \ge 2} \, q^{-1/2} \E(|X|^{q})^{1/q}\Big) \, \Big(\sup _{q \ge 2} \, q^{-1/2} \E(|Y|^{q})^{1/q} \Big)\\
& \le 2 \|X\|_{\psi_2} \, \|Y\|_{\psi_2} \,.
\end{align*}
Hence, by assumption $(ii)$, we obtain
\begin{align*}
\|v^{(ij)}_\ell\|_{\psi_1} &\le 2 \|\<\Sigma^{-1/2}_{i,\cdot}, \tilde{X}_\ell \>  \|_{\psi_2} \|\<\Sigma^{1/2}_{j,\cdot}, \tilde{X}_\ell\>\|_{\psi_2}\\
&\le 2 \|\Sigma^{-1/2}_{i,\cdot}\|_2  \|\Sigma^{1/2}_{j,\cdot}\|_2 \kappa^2 \le 2 \sqrt{C_{\max}/ C_{\min}}\, \kappa^2\,.
\end{align*}
Let $\kappa'  = 2 \sqrt{C_{\max} /C_{\min}} \kappa^2$. Applying Bernstein-type inequality for centered sub-exponential random variables~\cite{Vershynin-CS}, we get
\[
\prob\Big\{\frac{1}{n} \Big|\sum_{\ell=1}^n v^{(ij)}_\ell \Big| \ge \eps \Big\} \le 2 \exp \Big[ -\frac{n}{6} \min\Big((\frac{\eps}{e\kappa'})^2, \frac{\eps}{e\kappa'}\Big) \Big]\,.
\]
Choosing $\eps = a\sqrt{(\log p)/n}$, and 
assuming $n \ge [a/(e\kappa')]^2 \log p$, we arrive at
\[
\prob\bigg \{\frac{1}{n} \Big|\sum_{\ell=1}^n v^{(ij)}_\ell \Big| \ge a\sqrt{\frac{\log p}{n}}  \bigg\} 
\le 2 p^{-a^2/(6e^2\kappa'^2)}\,.
\]
The result follows by union bounding over all possible pairs $i, j \in [p]$.
\end{proof}

\subsection{Proof of Theorem \ref{thm:main_thm}}

Let
\begin{align}
\Delta_0 \equiv \Big(\frac{16ac\, \sigma}{C_{\rm min}}
\Big)\frac{s_0\log p}{\sqrt{n}} 
\end{align}
be a shorthand for the bound on $\|\Delta\|_{\infty}$ appearing in
Eq.~(\ref{eq:TailBoundQuant}). Then we have
\begin{align}
\prob\Big(\|\Delta\|_{\infty}\ge \Delta_0\Big)  \le & 
\prob\Big(\big\{\|\Delta\|_{\infty}\ge
\Delta_0\big\}\cap\event_n(C_{\min}^{1/2}/2,s_0,3/2)\cap \coev_n(a)\Big)\nonumber\\
&+\prob\big(\event_n^2(C_{\min}^{1/2}/2,s_0,3/2)\big)
+\prob\big(\coev_n^c(a)\big)\\
\le &
\prob\Big(\big\{\|\Delta\|_{\infty}\ge
\Delta_0\big\}\cap\event_n(C_{\min}^{1/2}/2,s_0,3/2)\cap \coev_n(a)\Big)
+4\, e^{-c_1n}+ 2\, p^{-c_2}\, ,
\end{align}
where, in the firsr equation ${\cal A}^c$ denotes the complement of
event ${\cal A}$ and  the second inequality follows from Theorem
\ref{thm:event_thm}. Notice, in particular, that the bound
(\ref{eq:BoundEvA}) can be applied for $K=3/2$ since, under the
present assumptions  $20\kappa^2\sqrt{(\log p)/n}\le 1/2$.

Finally 
\begin{align}
\prob\Big(\big\{\|\Delta\|_{\infty}\ge
\Delta_0\big\}\cap\event_n(C_{\min}^{1/2}/2,s_0,3/2)\cap \coev_n(a)\Big)&\nonumber\\
\le\sup_{\bX \in \event_n(C_{\min}^{1/2}/2,s_0,3/2)\cap \coev_n(a)}&
\prob\Big(\|\Delta\|_{\infty}\ge
\Delta_0\Big|\bX\Big)\le 2\, p^{-\tilde{c_0}}\, ,
\end{align}
Here the last inequality follows from Theorem \ref{thm:deterministic}
applied per given $\bX \in \event_n(C_{\min}^{1/2}/2,s_0,3/2)\cap
\coev_n(a)$
and hence using the bound (\ref{eq:FirstBoundDelta}) with $\phi_0 =
C_{\min}^{1/2}/2$, $K=3/2$, $\mu_* = a\sqrt{(\log p)/n}$.
%

\subsection{Proof of Lemma~\ref{lemma:LastDistribution}}
\label{proof:LastDistribution}

We will prove that, under the stated assumptions 
\begin{eqnarray}
\lim\sup_{n\to\infty}\sup_{ \|\theta_0\|_0 \le s_0 }\prob \left\{\frac{\sqrt{n}(\htheta^u_i - \theta_{0,i})}{\hsigma [M \hSigma M^\sT]_{i,i}^{1/2}} 
\le x  \right\} \le \Phi(x)  \, .\label{eq:LastDistrUB}
\end{eqnarray}
A matching lower bound follows by a completely analogous argument.

Notice that by Eq.~(\ref{eq:GeneralRepresentationBis}), we have
\begin{align}
\frac{\sqrt{n} (\htheta^u_i - \theta_{0,i})}{\sigma [M\hSigma M^\sT]_{ii}^{1/2}}
= \frac{e_i^\sT M\bX^\sT W}{\sigma [M\hSigma M^\sT]_{ii}^{1/2}} + \frac{\Delta_i}{\sigma [M\hSigma M^\sT]_{ii}^{1/2}}\,. 
\end{align}
Let $V = \bX M^\sT e_i/(\sigma [M\hSigma M^\sT]_{ii}^{1/2})$ and $\tZ \equiv V^\sT W$. We claim that $\tilde{Z}\sim \normal(0,1)$.
To see this, note that $\|V\|_2 = 1$, and $V$ and $W$ are independent.
Hence,
\begin{align}
\prob(\tZ \le x) = \E\{\prob(V^\sT W\le x |V )\} = \E\{\Phi(x) | V\} = \Phi(x)\,,
\end{align}
which proves our claim. 
In order to prove Eq.~(\ref{eq:LastDistrUB}), fix $\eps>0$ and write
\begin{align}
\prob \left(\frac{\sqrt{n}(\htheta^u_i - \theta_{0,i})}{\hsigma [M \hSigma M^\sT]_{i,i}^{1/2}} 
\le x  \right) &= \prob\left(\frac{\sigma}{\hsigma} \tZ +
\frac{\Delta_i}{\hsigma[M \hSigma M^\sT]_{i,i}^{1/2}}\ge x\right)\\
&\le  \prob\Big(\frac{\sigma}{\hsigma} \tZ \le
x+\eps\Big)+\prob\left(\frac{|\Delta_i|}{\hsigma[M \hSigma
  M^\sT]_{i,i}^{1/2}}\ge \eps\right)\\
&\le  \prob\Big(\tZ \le
x+2\eps+\eps|x|\Big)+\prob\left(\frac{|\Delta_i|}{\hsigma[M \hSigma
  M^\sT]_{i,i}^{1/2}}\ge \eps\right)
+\prob\Big(\Big|\frac{\hsigma}{\sigma}-1\Big|\ge \eps\Big)\, .
\end{align}
By taking the limit and using the assumption
(\ref{eq:ConsistencySigma}), we obtain
\begin{align}
\lim\sup_{n\to\infty}\sup_{ \|\theta_0\|_0 \le s_0 }\prob &\left(\frac{\sqrt{n}(\htheta^u_i - \theta_{0,i})}{\hsigma [M \hSigma M^\sT]_{i,i}^{1/2}} 
\le x  \right) \le \\
&\Phi(x+2\eps+\eps|x|) +\lim\sup_{n\to\infty}\sup_{\|\theta_0\|_0 \le s_0 }\prob\left(\frac{|\Delta_i|}{\hsigma[M \hSigma
  M^\sT]_{i,i}^{1/2}}\ge \eps\right)\, .\nonumber
\end{align}
Since $\eps>0$ is arbitrary, it  is therefore sufficient to show that
the limit on the right hand side vanishes for any $\eps>0$.
 
Note that $[M \hSigma
  M^\sT]_{i,i}\ge 1/(4\hSigma_{ii})$ for all $n$ large enough,
  by Lemma \ref{lem:missing_bound}, and since $\coh = a\sqrt{(\log
    p)/n}\to 0$ as $n,p\to\infty$. We have therefore
\begin{align}
\prob\left(\frac{|\Delta_i|}{\hsigma[M \hSigma
  M^\sT]_{i,i}^{1/2}}\ge \eps\right)&\le
\prob\Big(\frac{2}{\hsigma}\hSigma_{ii}^{1/2}\, |\Delta_i|\ge \eps\Big)\\
&\le \prob\Big(\frac{8}{\sigma}\, |\Delta_i|\ge
\eps\Big)+\prob\Big(\frac{\hsigma}{\sigma}\ge
2\Big)+\prob(\hSigma_{ii}\ge \sqrt{2})\, .
\end{align}
Note that   $\prob\big((\hsigma/\sigma)\ge
2\big)\to 0$ by assumption  
(\ref{eq:ConsistencySigma}), and $\prob(\hSigma_{ii}\ge \sqrt{2})\to
0$ by Theorem \ref{thm:event_thm}.$(b)$. Hence
\begin{align}
\lim\sup_{n\to\infty}\sup_{ \|\theta_0\|_0 \le s_0 }\prob\left(\frac{|\Delta_i|}{\hsigma[M \hSigma
  M^\sT]_{i,i}^{1/2}}\ge \eps\right)&\le\lim\sup_{n\to\infty}\sup_{ \|\theta_0\|_0 \le s_0 }\prob\Big(\|\Delta\|_{\infty}\ge
\frac{\eps\sigma}{8}\Big)\\
&\le \lim\sup_{n\to\infty}\big(4\, e^{-c_1n}+4\,
p^{-(\tilde{c}_0\wedge c_2)}\big) = 0\, ,
\end{align}
where the last inequality follows from Eq.~(\ref{eq:TailBoundQuant})
since $s_0 = o(\sqrt{n}/\log p)$ and hence $(16acs_0\log p)/(C_{\rm
  min}\sqrt{n})\le \eps/8$ for all $n$ large enough.

This completes the proof of Eq.~(\ref{eq:LastDistrUB}). The matching
lower bound follows by the same argument.
%

\subsection{Proof of Theorem~\ref{thm:error-power}}
\label{proof:error-power}

We begin with proving Eq.~\eqref{eq:typeI}.
Defining $Z_i \equiv \sqrt{n}(\htheta^u_i - \theta_{0,i})/ (\hsigma [M\hSigma M^\sT]_{i,i}^{1/2})$, we have 
\begin{align*}
\lim_{n\to \infty}  \alpha_{i,n}(\hT) &= \lim_{n\to \infty} \sup_{\theta_0}
\Big\{\prob(P_i \le \alpha): i\in [p],\, \|\theta_0\|_0 \le s_0,\, \theta_{0,i} = 0 \Big\}\\
&= \lim_{n\to \infty} \sup_{\theta_0} \Big\{\prob\Big(\Phi^{-1}(1-\frac{\alpha}{2}) \le \frac{\sqrt{n} |\htheta^u_i|}{\hsigma [M\hSigma M^\sT]_{i,i}^{1/2}} \Big):
\,i\in [p],\, \|\theta_0\|_0 \le s_0,\, \theta_{0,i} = 0 \Big\}\\
& = \lim_{n\to \infty} \sup_{\theta_0} \Big\{ \prob\Big(\Phi^{-1}(1-\frac{\alpha}{2}) \le |Z_i| \Big):
\,i\in [p],\, \|\theta_0\|_0 \le s_0 \Big\} \le \alpha\,,
\end{align*}
where the last inequality follows from Lemma \ref{lemma:LastDistribution}.

We next prove Eq.~\eqref{eq:power}. 
Recall that $\Sigma^{-1}_{\cdot,i}$ is a feasible solution
of~\eqref{eq:optimization}, for $1\le i\le p$ 
with probability at least $1-2p^{-c_2}$, as per Lemma~\ref{lem:main_lem}).
On this event, letting $m_i$ be the solution of the optimization
problem~\eqref{eq:optimization}, we have
\begin{align}
m_i^\sT \hSigma m_i &\le\Sigma^{-1}_{i,\cdot} \hSigma \Sigma^{-1}_{\cdot,i} \nonumber \\
& = (\Sigma^{-1}_{i,\cdot} \hSigma \Sigma^{-1}_{\cdot,i} -
\Sigma_{ii}^{-1})+ \Sigma^{-1}_{i,i} \nonumber \\
& = \frac{1}{n} \sum_{j=1}^N(V_j^2-
\Sigma_{ii}^{-1})+ \Sigma^{-1}_{i,i} \,,
\end{align}
where $V_j = \Sigma^{-1}_{i,\cdot}X_j$ are i.i.d. with $\E(V_j^2) =
\Sigma^{-1}_{ii}$ and sub-gaussian norm 
$\|V_j\|_{\psi_2} \le \|\Sigma_{i,\cdot}^{-1/2}\|_2\|
\Sigma^{-1/2} X_j\|_{\psi_2}\le \kappa
\sqrt{\Sigma_{i,i}^{-1}}$. Letting $U_j=V_j^2-
\Sigma_{ii}^{-1}$, we have that $U_j$ is zero mean and 
sub-exponential with $\|U_j\|_{\psi_1}\le 2 \|V_j^2\|_{\psi_1} \le 2 \|V_j\|_{\psi_2}^2\le
2\kappa^2\Sigma_{ii}^{-1}\le 2\kappa^2\sigma_{\rm min}(\Sigma)^{-1}\le
2\kappa^2C_{\rm min}^{-1}\equiv \kappa'$. 
Hence, by applying Bernstein inequality (as, for instance, in the
proof of Lemma \ref{lem:main_lem}), we have, 
for $\eps\le e\kappa'$,
\begin{align}
\prob\Big(m_i^\sT \hSigma m_i \ge \Sigma^{-1}_{i,i}+\eps\Big)\le
2\, e^{-(n/6)(\eps/e\kappa')^2}+2\, p^{-c_2}\, .
\end{align}
Therefore, by Borel-Cantelli (since we can make $c_2\ge 2$ by a
suitable choice of $a$), we have, almost surely
\begin{align}
\lim\sup_{n\to\infty} [m_i^\sT \hSigma m_i -\Sigma^{-1}_{i,i}]\le 0\, .
 \label{eq:sigmaiB}
\end{align}

This bound leads to a lower bound for the power.
First of all, a straightforward manipulation yields
as follows, letting $z_* \equiv \Phi^{-1}(1-\alpha/2)$:
\begin{align*}
&\lim\inf_{n\to \infty} \frac{1-\beta_{i,n}(\hT;\lb)}{1-\beta_{i,n}^{*}(\lb)} \\
&= \lim\inf_{n\to \infty} \frac{1}{1-\beta_i^{*}(\lb;n)} \inf_{\theta_0} 
\Big\{\prob(P_i \le \alpha) : \, \|\theta_0\|_0 \le s_0,\, |\theta_{0,i}| \ge \lb \Big\}\\
&= \lim\inf_{n\to \infty} \frac{1}{1-\beta_{i,n}^{*}(\lb)}  \inf_{\theta_0} \Big\{\prob\Big(z_* \le \frac{\sqrt{n} |\htheta^u_i|}{\hsigma [M \hSigma M^\sT]_{i,i}^{1/2}} \Big): \|\theta_0\|_0 \le s_0,\,|\theta_{0,i}| \ge \lb \Big\}\\
&=  \lim\inf_{n\to \infty} \frac{1}{1-\beta_{i,n}^{*}(\lb)} \inf_{\theta_0} \Big\{ \prob\Big(z_*\le \Big|Z_i + \frac{\sqrt{n} \theta_{0,i}}
{\hsigma  [M \hSigma M^\sT]_{i,i}^{1/2}}\Big| \Big):
\|\theta_0\|_0 \le s_0,\,|\theta_{0,i}| \ge \lb \Big\}\\
&\stackrel{(a)}{\ge} \lim\inf_{n\to \infty} \frac{1}{1-\beta_{i,n}^{*}(\lb)} \inf_{\theta_0} \Big\{ \prob\Big(z_*\le \Big|Z_i +
 \frac{\sqrt{n} \lb}{\sigma 
[\Sigma_{i,i}^{-1}]^{1/2}}\Big|\Big): \|\theta_0\|_0 \le s_0 \Big\}\\
&=  \lim\inf_{n\to \infty} \frac{1}{1-\beta_{i,n}^{*}(\lb)} \Big\{1- \Phi\Big(z_*- \frac{\sqrt{n} \lb}{\sigma [\Sigma_{i,i}^{-1}]^{1/2}}\Big) 
+ \Phi\Big(-z_*- \frac{\sqrt{n} \lb}{\sigma [\Sigma_{i,i}^{-1}]^{1/2}}\Big) \Big\}\\
&= \lim\inf_{n\to \infty} \frac{1}{1-\beta_{i,n}^{*}(\lb)} G\Big(\alpha, \frac{\sqrt{n} \lb}{\sigma  [\Sigma_{i,i}^{-1}]^{1/2}} \Big) = 1\,.
\end{align*}
Here $(a)$ follows from Eq.~\eqref{eq:sigmaiB} and the fact $|\theta_{0,i}| \ge \lb$. 

\subsection{Proof of Theorem~\ref{thm:nongauss}}
\label{proof:nongauss}
Under the assumptions of Theorem~\ref{thm:main_thm} and assuming $s_0 = o(\sqrt{n}/\log p)$, we have
\[
\sqrt{n}(\htheta^u - \theta_{0})  = \frac{1}{\sqrt{n}} M \bX^\sT W + \Delta\,
\]
with $\|\Delta\|_\infty = o(1)$. Using Lemma~\ref{lem:missing_bound}, we have
\[
\frac{\sqrt{n}(\htheta^u_i - \theta_{0,i}) }{\sigma [M\hSigma M^\sT]_{i,i}^{1/2}} = Z_i + o(1)\,,
\quad \text{with } \, Z_i \equiv \frac{1}{\sqrt{n}}\frac{m_i^\sT \bX^\sT W}{\sigma [m_i^\sT \hSigma m_i]^{1/2}} \,.
\]

The following lemma characterizes the limiting distribution of
$Z_i|\bX$ which implies the validity of the proposed $p$-value 
$P_i$ and confidence intervals.

\begin{lemma}\label{lem:nongauss}
Suppose that the noise variables $W_i$ are independent with $\E(W_i) =
0$, and $\E(W_i^2) = \sigma^2$,
 and $E(|W_i|^{2+a})\le C\, \sigma^{2+a}$
for some $a>(1/2-\beta)^{-1}$.
Let $M = (m_1,\dotsc,m_p)^\sT$ be the matrix with rows $m_i^\sT$ obtained by solving optimization 
problem~\eqref{eq:optimization_mod}. For $i \in [p]$, define
\[
Z_i = \frac{1}{\sqrt{n}}\frac{m_i^\sT \bX^\sT W}{\sigma [m_i^\sT \hSigma m_i]^{1/2}} \,.
\]
Under the assumptions of Theorem~\ref{thm:main_thm}, for any sequence $i = i(n) \in [p]$, and 
any $x\in \reals$, we have 
\[
\lim_{n\to \infty} \prob(Z_i \le x|\bX) = \Phi(x)\,. 
\]
\end{lemma}
Lemma~\ref{lem:nongauss} is proved in Appendix~\ref{app:nongauss}.

\subsubsection*{Acknowledgments}

A.J. is supported by a Caroline and Fabian
Pease Stanford Graduate Fellowship.
This work was partially supported by the NSF CAREER award CCF-0743978, the NSF
grant DMS-0806211, and the grants AFOSR/DARPA FA9550-12-1-0411 and FA9550-13-1-0036.

\newpage
\appendix

\section{Proof of technical lemmas}

\subsection{Proof of Lemma~\ref{lem:missing_bound}}
\label{app:missing_bound}

Let $C_{i}(\coh)$ be the optimal value of the optimization problem~\eqref{eq:optimization}.
We claim that 
\begin{align}
C_i(\coh)\ge \frac{(1-\coh)^2}{\hSigma_{ii}}\, .\label{eq:ClaimMM}
\end{align}
To prove this claim notice that the constraint implies (by considering
its $i$-th component):
\begin{align*}
1-\<e_i,\hSigma m\>\le \coh\,.
\end{align*}
Therefore if $\tilde{m}$ is feasible and $c\ge 0$, then
\[
\<\tilde{m}, \hSigma \tilde{m}\> \ge \<\tilde{m},\hSigma
\tilde{m}\>+c(1-\coh)-c\<e_i,\hSigma \tilde{m}\>
\ge \min_m\Big\{\<m,\hSigma
m\>+c(1-\coh)-c\<e_i,\hSigma m\>\Big\}\, .
\]
Minimizing over all feasible $\tilde{m}$ gives
\begin{align}
C_{i}(\coh)\ge\min_m\Big\{\<m,\hSigma
m\>+c(1-\coh)-c\<e_i,\hSigma m\>\Big\}\, .
\end{align}
The minimum over $m$ is achieved at $m=ce_i/2$. Plugging in for $m$, we get
\begin{align}
C_{i}(\coh)\ge c(1-\coh)-\frac{c^2}{4}\hSigma_{ii}
\end{align}
Optimizing this bound over $c$, we obtain the claim
(\ref{eq:ClaimMM}), with the optimal choice being $c = 2(1-\coh)/\hSigma_{ii}$.

\subsection{Proof of Lemma~\ref{lem:nongauss}}
\label{app:nongauss}
Write 
\[
Z_i  = \frac{1}{\sqrt{n}} \sum_{j=1}^n \xi_j\, \quad\quad \text{with}\quad \xi_j \equiv 
\frac{m_i^\sT X_j W_j}{\sigma [m_i^\sT \hSigma m_i]^{1/2}}\,.
\]
Conditional on $\bX$, the summands $\xi_j$ are zero mean and independent. Furthermore, $\sum_{j=1}^n \E(\xi_j^2|\bX) =n$.
We next prove the Lindenberg condition as per
Eq.~\eqref{eq:lindenberg}. Let $c_n \equiv (m_i^\sT \hSigma
m_i)^{1/2}$. 
By Lemma \ref{lem:missing_bound}, we have, almost surely,
$\lim\inf_{n\to\infty}c_n\ge c_{\infty}>0$.
If all the optimization problems in~\eqref{eq:optimization_mod}
are feasible, then $|\xi_j| \le  c_n^{-1} \|\bX m_i\|_{\infty} \|W\|_\infty/\sigma \le c_n^{-1} n^{\beta} (\|W\|_\infty/\sigma)$. Hence,
\begin{align*}
\lim_{n\to \infty} \frac{1}{n} \sum_{j=1}^n \E \Big(\xi_j^2 \ind_{\{|\xi_j| > \eps \sqrt{n}\}} |\bX \Big) 
&\le \lim_{n\to \infty} \frac{1}{n} \sum_{j=1}^n \E \Big(\xi_j^2 \ind_{\{\|W\|_\infty/\sigma > \eps c_n n^{1/2-\beta} \}} |\bX \Big)\\
& = \lim_{n\to \infty}  \frac{1}{n} \sum_{j=1}^n 
\frac{m_i^\sT X_j X_j^\sT m_i}{m_i^\sT \hSigma m_i}
\E(\tW_j^2\, \ind_{\{\|\tW\|_\infty > \eps c_{\infty}
  n^{1/2-\beta} \}}  \Big)\\
&\le \lim_{n\to\infty}n \E(\tW_1^2\, \ind_{\{|\tW_1| > \eps c_{\infty}
  n^{1/2-\beta} \}}  \Big) \\
&\le c'(\eps)\lim_{n\to\infty} n^{1-a(1/2-\beta)}\E\{|\tW_1|^{2+a}\} =
0\, .
\end{align*}
where $\tW_j=W_j/\sigma$ and the last limit follows by taking
$a>(1/2-\beta)^{-1}$ as per the assumptions.

Using Lindenberg central limit theorem, we obtain $Z_i |\bX$ converges
weakly to standard normal distribution, and hence,
$\bX$-almost surely
\[
\lim_{n\to \infty} \prob(Z_i \le x|\bX) = \Phi(x)\,. 
\]

What remains is to show that with high probability all the $p$ optimization problems in~\eqref{eq:optimization_mod} are feasible.
In particular, we show that $\Sigma^{-1}_{i,\cdot}$ is a feasible solution to the $i$-th optimization problem, for $i \in [p]$. 
By Lemma~\ref{lem:main_lem}, $|\Sigma^{-1} \hSigma - \id|_{\infty} \le \coh$, with high probability. Moreover,
\begin{align*}
\sup_{j\in [p]} \|\Sigma^{-1}_{i,\cdot} X_j\|_{\psi_2} &= \sup_{j\in [p]} \|\Sigma^{-1/2}_{i,\cdot} \Sigma^{-1/2} X_j\|_{\psi_2}\\
&= \|\Sigma^{-1/2}_{i,\cdot}\|_2 \sup_{j\in [p]} \|\Sigma^{-1/2} X_j\|_{\psi_2}\\
 &= [\Sigma^{-1}_{i,i}]^{1/2}  \sup_{j\in [p]} \|\Sigma^{-1/2} X_j\|_{\psi_2} = O(1)\,.
\end{align*}
Using tail bound for sub-gaussian variables $\Sigma^{-1}_{i,\cdot} X_j$ and union bounding over $j \in [n]$, we get
\[
\prob(\|\bX\Sigma^{-1}_{\cdot,i}\|_\infty > n^\beta) \le n e^{-cn^{2\beta}}\,,
\]
for some constant $c > 0$.
Note that $s_0 = o(\sqrt{n}/\log p)$ implies $p=e^{o(n^{2\beta})}$.
 Hence, eventually almost surely, $\Sigma^{-1}_{i,\cdot}$ is a feasible solution to optimization problem~\eqref{eq:optimization_mod},
for all $i \in [p]$.

%
%
\section{Corollaries of Theorem \ref{thm:main_thm}}

\subsection{Proof of Corollary \ref{coro:UIsNotBiased}}\label{app:UIsNotBiased}

By Theorem \ref{thm:deterministic}, for any $\bX\in
\event_n(\sqrt{C_{\min}}/2,s_0,3/2)\cap \coev_n(a)$, we have
\begin{align}
\prob\left\{\|\Delta\|_{\infty}\ge L\, c\Big|\bX\right\} \le 2\,
p^{1-(c^2/48)}\, ,\;\;\;\; L\equiv \frac{16a\sigma}{C_{\min}}\,
\frac{s_0\log p}{\sqrt{n}}\, .
\end{align}
(This is obtained by setting $\phi_0 = C_{\min}^{1/2}/2$, $K=3/2$,
$\mu_* = a\sqrt{(\log p)/n}$ in Eq.~(\ref{eq:FirstBoundDelta}). 
Hence
\begin{align}
\|\bias(\htheta^u)\|_{\infty}&\le
\frac{1}{\sqrt{n}}\E\big\{\|\Delta\|_{\infty}\big|\bX\big\}\\
&= \frac{L'}{\sqrt{n}}\int_{0}^\infty
\prob\left\{\|\Delta\|_{\infty}\ge B'\, c\Big|\bX\right\} \, \de c\\
& \le  \frac{L}{\sqrt{n}}\int_{0}^{\infty} \min(1 , p^{1-(c^2/48)})
\,\de c \le \frac{10 L}{\sqrt{n}}\, ,
\end{align}
which coincides with Eq.~(\ref{eq:BoundCoroPerX}). 
The probability estimate  (\ref{eq:CoroProbabilityBound}) simply
follows from Theorem \ref{thm:event_thm} using union bound. 

\subsection{Proof of Corollary
  \ref{coro:LASSOIsBiased}}\label{app:LASSOIsBiased}

By Theorem \ref{thm:event_thm}.$(a)$, we have (setting $C_{\min} =
C_{\max} = \kappa = 1$):
\begin{align}
\prob\Big(\bX\in\cE_n(1/2,s_0,3/2)\Big)\ge 1-4\, e^{-n/c_*}\, .
\end{align}
Further, by Lemma \ref{lem:main_lem}, with 
$\hSigma \equiv \bX^{\sT}\bX/n$,
we have
\begin{align}
\prob\Big(\mu_*(\bX;\id)\le 30\sqrt{\frac{\log p}{n}}\Big)\ge 1-2\, p^{-3}\, .
\end{align}
Finally, by an obvious consequence of the proof of Theorem
\ref{thm:event_thm}.$(a)$
\begin{align}
\prob\Big(\Big\{\bX:\, \min_{i\in[p]}\hSigma_{ii}\ge \frac{1}{2}\Big\}\Big)\ge 1-2\, e^{-n/c_*}\, .
\end{align}

Hence, defining
\begin{align}
\cB_n \equiv \cE_n(1/2,s_0,3/2)\cap \Big\{\bX\in\reals^{n\times p}:\;
\mu_*(\bX;\id)\le 30\sqrt{\frac{\log p}{n}}\Big\}\cap
\Big\{\bX:\, \min_{i\in[p]}\hSigma_{ii}\ge \frac{1}{2}\Big\}
\end{align} 
we have the desired probability bound (\ref{eq:CoroPbound2}). 

Let
\begin{align}
\htheta^* \equiv \htheta^n + \frac{1}{n}\bX^{\sT}(Y-\bX\htheta^n)\, ,
\end{align}
where $\htheta^n = \htheta^n(Y,\bX;\lambda)$ is the LASSO solution
with $\lambda = \sigma\sqrt{(c^2\log p)/n}$.
By Theorem \ref{thm:deterministic}, we have, for any $\bX\in \cB_n$
\begin{align}
\htheta^* =\theta_0+ \frac{1}{\sqrt{n}}Z + \frac{1}{\sqrt{n}}\Delta\,, \quad 
Z|\bX \sim \normal(0,\sigma^2 \hSigma)\,, 
\end{align}
and further 
\begin{align}
\prob\Big\{\|\Delta\|_{\infty}\ge \frac{240 c\sigma s_0\log
  p}{\sqrt{n}}\Big|\bX\Big\}\le 2p^{1-(c^2/48)}\, ,
\end{align}
whence, proceeding as in the proof in the last section,
we get, for some universal numerical constant $c_{**}$, 
\begin{align}
\|\bias(\htheta^u)\|_{\infty}\le \frac{1}{\sqrt{n}}\E\Big\{\|\Delta\|_{\infty}\Big|\bX\Big\}\le c_{**}\sigma\frac{s_0\log
p}{n}\, . \label{eq:BoundDeltaStandard}
\end{align}

Next by Eq.~(\ref{eq:TwoBiasRelation}) we have 
\begin{align}
\big\|\bias(\htheta^n)\big\|_{\infty} \ge
\Big|\lambda\big\|\E\{v(\htheta^n)|\bX\}\big\|_{\infty}-
\big\|\bias(\htheta^u)\big\|_{\infty}\Big|\, .
\end{align}
Hence, in order to prove Eq.~(\ref{eq:CoroBiasStd}), it is sufficient
to prove that $\|\E\{v(\htheta^n)|\bX\}\|_{\infty}\ge 2/3$.

Note that $v(\htheta^n)_i = 1$ whenever $\htheta^n_i>0$ and,
and $|v(\htheta^n)_i|\le 1$, and therefore (letting $b_0 \equiv c_{**}\sigma(s_0\log
p)/n$)
\begin{align}
1-\E\{v(\htheta^n)_i|\bX\}&\le 2\prob\Big(\htheta^n_i\le 0\Big|\bX\Big) 
\le 2\prob\Big(\htheta^n_i\le \lambda\Big|\bX\Big)\\
& \le 2\prob\Big(\theta_{0,i}+\frac{1}{\sqrt{n}}Z_i
+\frac{1}{\sqrt{n}}\Delta_i\le \lambda\Big|\bX\Big)\\
& \le 2\prob\Big(\frac{1}{\sqrt{n}}Z_i\le
\lambda+b_0-\theta_{0,i}\Big|\bX\Big)\\
& = 2\Phi\Big((\lambda+b_0-\theta_{0,i})\sqrt{n/(\sigma^2\hSigma_{ii})}\Big)\\
& \le 2\Phi\Big((\lambda+b_0-\theta_{0,i})\sqrt{2n/(3\sigma^2)}\Big)
\end{align}
with $\Phi(x)$ the standard normal distribution function, and in the
last inequality we used the fact that $\max_{i\in [p]}\hSigma_{ii}\le
3/2$ on $\cB_n$.
We then choose $\theta_0$ so that $\theta_{0,i}\ge
b_0+\lambda+\sqrt{30\sigma^2/n}$, for $i\in [p]$ in the support of
$\theta_0$.
We therefore obtain 
\begin{align}
\E\{v(\htheta^n)_i|\bX\}\ge 1-2\Phi(-\sqrt{20})\ge \frac{2}{3}\, .
\end{align}
This finishes the proof of Eq.~(\ref{eq:CoroBiasStd}).
Equations (\ref{eq:CoroBiasStd_final}) and
(\ref{eq:CoroBiasStd_final2})  are obtained by substituting 
$\lambda = c\sigma\sqrt{(\log p)/n}$ and using
Eq.~(\ref{eq:CoroBiasStd}).
%
%
\section{Proof of Lemma \ref{lemma:ConsistencySigma}}
\label{sec:ConsistencySigma}

Let $\event_n = \event_n(\phi_0,s_0,K)$ be the event defined as per
Theorem  \ref{thm:event_thm}.$(a)$. In particular, we take 
$\phi_0 = C_{\rm min}^{1/2}/2$, and $K\ge
1+20\kappa^2\sqrt{(\log p)/n}$ (for, instance $K=1.1$ will work for
all $n$ large enough since $(s_0\log p)^2/n\to 0$, with $s_0\ge 1$, by
assumption). 
Further note that we can assume without loss of generality
$n\ge\nu_0\,  s_0\log (p/s_0)$,
since $s_0 = o(\sqrt{n}/\log p)$.
Fixing $\eps>0$, we have therefore
\begin{align}
\prob\Big(\Big|\frac{\hsigma}{\sigma}-1\Big|\ge \eps\Big)& \le 
\sup_{\bX\in\event_n}\prob\Big(\Big|\frac{\hsigma}{\sigma}-1\Big|\ge
\eps\, \Big|\, \bX\,\Big) + \prob\big(\bX\not\in\event_n\big)\\
& \le \sup_{\bX\in\event_n}\prob\Big(\Big|\frac{\hsigma}{\sigma}-1\Big|\ge
\eps\, \Big|\, \bX\,\Big) + 4\, e^{-c_1n}\, ,
\end{align}
where $c_1>0$ is a constant defined as per Theorem
\ref{thm:event_thm}.$(a)$. 

We are therefore left with the task of bounding the first term in the
last expression above, uniformly over $\theta_0\in\reals^p$,
$\|\theta_0\|_0\le s_0$. For $\bX\in\event_n$, we can apply
\cite[Theorem 1.2]{SZ-scaledLASSO} whereby (using the notations of
\cite{SZ-scaledLASSO}, with their $\lambda_0$ replaced by $\tlambda$) $\xi = 3$, $T=\supp(\theta_0)$,
$\kappa(\xi,T)\ge \phi_0$, $\eta_*(\tlambda,\xi)\le
4s_0\tlambda^2/\phi_0^2$.
By a straightforward manipulation of Eq.~(13) in
\cite{SZ-scaledLASSO}, we have, for $\bX\in\event$, and 
$\|\bX^{\sT}W/n\|_{\infty}\le \tlambda/4$
(letting $\sigma^*$ the
oracle estimator of $\sigma$ introduced there)
\begin{align}
\Big|\frac{\hsigma}{\sigma^*}-1\Big|\le \frac{4\sqrt{s_0}\tlambda}{\phi_0}\le
\frac{\eps}{2}
\end{align}
where the last inequality follows for all $n$ large enough since $s_0
= o(\sqrt{n}/\log p)$.

Hence 
\begin{align}
\sup_{\bX\in\event_n}\prob\Big(\Big|\frac{\hsigma}{\sigma}-1\Big|\ge
\eps\, \Big|\, \bX\,\Big)\le \sup_{\bX\in\event_n}\prob\Big(\, \|\bX^{\sT}W/n\|_{\infty}\le \tlambda/4\Big|\, \bX\,\Big)
+\sup_{\bX\in\event_n}\prob\Big(\Big|\frac{\sigma^*}{\sigma}-1\Big|\ge
\frac{1}{10}\, \Big|\, \bX\,\Big) \, ,
\end{align} 
where we note that the right hand side is independent of $\theta_0$.
The first term vanishes as $n\to\infty$ by a standard tail bound on the
supremum of $p$ Gaussian random variables. The second term also
vanishes because it is controlled by the tail of a chi-squared random
variable \cite{SZ-scaledLASSO}.
%
%
\section{Proof of Theorem \ref{thm:FWER}}
\label{app:FWER}

By definition, letting $\cF_{p,s_0} \equiv\{x\in\reals^p:\, \|x\|_0\le
s_0\}$, and fixing $\eps\in (0,1/10)$
\begin{align}
\FWER(\hTf,n) &=\sup_{\theta_0\in\cF_{p,s_0}} \prob\left\{\exists i\in
    [p]\setminus\supp(\theta_0) \mbox{ s.t. } \frac{\sqrt{n}\,
      |\htheta^u_i-\theta_{0,i}|}{\hsigma [M \hSigma
      M^\sT]_{i,i}^{1/2}}  \ge
\Phi^{-1}\Big(1-\frac{\alpha}{2p}\Big)\right\}\\
& \le \sup_{\theta_0\in\cF_{p,s_0}} \prob\left\{\exists i\in
    [p]\setminus\supp(\theta_0) \mbox{ s.t. } \frac{\sqrt{n}\,
      |\htheta^u_i-\theta_{0,i}|}{\sigma [M \hSigma
      M^\sT]_{i,i}^{1/2}}  \ge (1-\eps)
\Phi^{-1}\Big(1-\frac{\alpha}{2p}\Big)\right\}\\
&\nonumber\phantom{\le\le}+
\sup_{\theta_0\in\cF_{p,s_0}}
\prob\Big(\Big|\frac{\hsigma}{\sigma}-1\Big|\ge\frac{\eps}{2}\Big)\, .
\end{align}

Since the second term vanishes as $n\to\infty$ by assumption
Eq.~(\ref{eq:ConsistencySigma}), 
it is sufficient to consider the first term.  Using Bonferroni
inequality, letting $z_\alpha(\eps) \equiv  (1-\eps)
\Phi^{-1}\big(1-\frac{\alpha}{2p} \big)$, we have
\begin{align}
\underset{n\to\infty}{\lim\sup}\,\, \FWER(\hTf,n) &\le \underset{n\to\infty}{\lim\sup}\,\,\sum_{i=1}^p
\sup_{\theta_0\in\cF_{p,s_0},\theta_{0,i}=0}
\prob\left\{ \frac{\sqrt{n}\,
      |\htheta^u_i-\theta_{0,i}|}{\sigma [M \hSigma
      M^\sT]_{i,i}^{1/2}}  \ge z_\alpha(\eps)
\right\}\\
&=  \underset{n\to\infty}{\lim\sup}\,\,\sum_{i=1}^p
\sup_{\theta_0\in\cF_{p,s_0},\theta_{0,i}=0}
\prob\left\{ \left|\tZ_i+\frac{\Delta_i}{\sigma[M\hSigma
      M^{\sT}]_{ii}^{1/2}}\right|\ge z_\alpha(\eps)\right\}
\end{align}
where, by Theorem \ref{thm:main_thm}, $\tZ_i\sim\normal(0,1)$
and $\Delta_i$ is given by Eq.~(\ref{eq:GeneralRepresentationBis}). We
then have
\begin{align}
\underset{n\to\infty}{\lim\sup}\,\,\FWER(\hTf,n) &\le  \underset{n\to\infty}{\lim\sup}\,\,\sum_{i=1}^p
\prob\big\{ |\tZ_i|\ge z_\alpha(\eps)-\eps\big\} \nonumber\\
&\,\,\,+\underset{n\to\infty}{\lim\sup}\,\,  \sum_{i=1}^p\sup_{\theta_0\in\cF_{p,s_0},\theta_{0,i}=0}
\prob\left\{\|\Delta\|_{\infty}\ge
  \frac{\eps\sigma}{2\hSigma_{ii}^{1/2}}\right\}\nonumber\\
&\le  2\big(1-\Phi(z_{\alpha}(\eps)-\eps)\big) +\underset{n\to\infty}{\lim\sup}\,\,
p\max_{i\in [p]}\prob(\hSigma_{ii}\ge 2)\nonumber\\ 
&\,\,\,+\underset{n\to\infty}{\lim\sup}\,\,  \sup_{\theta_0\in\cF_{p,s_0},\theta_{0,i}=0}\,
p\, 
\prob\left\{\|\Delta\|_{\infty}\ge
  \frac{\eps \sigma}{4}\right\}\, ,\label{eq:BonferroniLast}
\end{align}
where in the first inequality, we used 
 $[M \hSigma M^\sT]_{i,i}\ge 1/(4\hSigma_{ii})$ for all $n$ large enough,
  by Lemma \ref{lem:missing_bound}, and since $\coh = a\sqrt{(\log
    p)/n}\to 0$ as $n,p\to\infty$.  Now the second term in
  the right hand side of Eq.~(\ref{eq:BonferroniLast})
vanishes by Theorem  \ref{thm:event_thm}.$(a)$, and the last term is zero by Theorem \ref{thm:main_thm} 
since $\sqrt{n}/\log(p)\ge s_0\ge 1$.
Therefore 
\begin{align}
\lim\sup_{n\to\infty}\FWER(\hTf,n) \le
2\big(1-\Phi(z_{\alpha}(\eps)-\eps)\big) \, ,
\end{align}
and the claim follows by letting $\eps\to 0$. 
%

\bibliographystyle{amsalpha}

\newcommand{\etalchar}[1]{$^{#1}$}
\providecommand{\bysame}{\leavevmode\hbox to3em{\hrulefill}\thinspace}
\providecommand{\MR}{\relax\ifhmode\unskip\space\fi MR }
\providecommand{\MRhref}[2]{%
  \href{http://www.ams.org/mathscinet-getitem?mr=#1}{#2}
}
\providecommand{\href}[2]{#2}

\end{document}